\documentclass[12pt]{article}
\usepackage{epsfig}
\usepackage{amsfonts}
\usepackage{amscd}
\usepackage{latexsym}
\usepackage{amsmath,amssymb,amsthm}
\usepackage{verbatim}
\usepackage{setspace}
\usepackage{color}
\usepackage{cite}
\usepackage{graphicx}
\usepackage{mathtools}
\usepackage[all]{xy}
\usepackage{tikz}

\usepackage[textheight=9in, textwidth=6.5in, letterpaper]{geometry}
\usepackage{color}   %May be necessary if you want to color links
\usepackage{hyperref}
\hypersetup{
    colorlinks=true,  %set true if you want colored links
    linktoc=all,     %set to all if you want both sections and subsections linked
    linkcolor=black,  %Choose some color if you want links to stand out
    citecolor=black,
    filecolor=black,
    urlcolor=black,
}

\numberwithin{equation}{section}

\newtheorem{Theorem}{Theorem}[section]
\newtheorem*{Theorem*}{Theorem}

\newtheorem{Proposition}[Theorem]{Proposition}

 { \theoremstyle{definition}
\newtheorem{Definition}[Theorem]{Definition}

\newtheorem{Example}[Theorem]{Example}
\newtheorem{Remark}[Theorem]{Remark} }
\def\ve{\epsilon} % for this paper only
\def\vp{\varphi}

\def\p{\partial}

\def\<{\langle}
\def\>{\rangle}

\def\cO{\mathcal{O}}

\def\be{\begin{equation}}
\def\ee{\end{equation}}
\def\beq{\be\begin{array}{c}}
\def\eeq{\end{array}\ee}
\def\bes{\be\begin{split}}
\def\ees{\end{split} \ee}
\def\bs{\begin{split}}
\def\es{\end{split} }
\def\nn{\nonumber}
\def\b{{\beta}}
\def\a{{\alpha}}
\def\g{{ \gamma}}

\def\e{{\epsilon}}

   \makeatletter
  \let\over=\@@over \let\overwithdelims=\@@overwithdelims
  \let\atop=\@@atop \let\atopwithdelims=\@@atopwithdelims
  \let\above=\@@above \let\abovewithdelims=\@@abovewithdelims
\renewcommand\section{\@startsection {section}{1}{\z@}%
                                   {-3.5ex \@plus -1ex \@minus -.2ex}%nn
                                   {2.3ex \@plus.2ex}%
                                   {\normalfont\large\bfseries}}

\renewcommand\subsection{\@startsection{subsection}{2}{\z@}%
                                     {-3.25ex\@plus -1ex \@minus -.2ex}%
                                     {1.5ex \@plus .2ex}%
                                     {\normalfont\bfseries}}

\begin{document}
\begin{titlepage}
\unitlength = 1mm

\vskip 1cm
\begin{center}

{\LARGE {\textsc{Landau-Ginzburg-Saito theory for descendant Gromov-Witten theory on projective line}}} \\

\vspace{1cm}
 Vyacheslav Lysov\\
\vspace{1cm}
{\it Center for Mathematics and Interdisciplinary Sciences, Fudan University, Shanghai, 200433, China}\\
{\it Shanghai Institute for Mathematics and Interdisciplinary Sciences (SIMIS), Shanghai, 200433, China}
\vspace{0.3cm}
\vspace{1cm}
\vspace{1cm}
\begin{abstract}
We define the correlation functions for the descendants in the Landau-Ginzburg-Saito theory. We show that the correlation functions obey puncture, divisor, dilaton, and topological recursion relations. We formulate the map between the descendant observables in the GW theory on the projective line and the descendant observables in the mirror LGS theory. We prove that the LGS correlation functions of the mirror observables are equal to the GW invariants with descendants. 
\end{abstract}%\vspace{0.5cm}
\vspace{1.0cm}
\end{center}
\end{titlepage}

\pagestyle{empty}
\pagestyle{plain}

\pagenumbering{arabic}
\tableofcontents

\section{Introduction}

Mirror symmetry in mathematical physics is a map between the Gromov-Witten (GW) invariants in the A-model and correlators in the Landau-Ginzburg (LG) cohomological field theory in the B-model. The usual perspective is that the GW invariants are incredibly complicated to evaluate, while the LG correlation functions are much easier to evaluate. The correlation functions are essentially multi-dimensional residue integrals. Hence, we naturally expect to have simple, explicit examples realizing this perspective. 

The simplest example is the mirror symmetry for the projective line $\mathbb{P}^1$. The GW invariants for $\mathbb{P}^1$ are known in explicit form \cite{dubrovin2019gromov} in full generality. Namely, gravitational descendants at arbitrary genus are included. Moreover, we know many different ways of evaluating them, including the Toda conjecture \cite{pandharipande2000toda}, Hurwitz numbers relation \cite{okounkov2006gromov}, integrability \cite{dubrovin2019gromov}, and many others. The mirror LG model has target space $\mathbb{C}^{\ast}$, holomorphic top form $\Omega = dY$ and simple mirror superpotential $W = e^{iY} + q e^{-iY}$. Hence, we expect a simple B-model evaluation of the invariants in the same generality via the residue-type integrals. 

Unfortunately, no such B-model construction in full generality is available in the literature. We only know the partial results. The mirror map for $\mathbb{P}^1$ is extremely easy due to the simplicity of the GW invariants on $\mathbb{P}^1$. The higher genus invariants without gravitational descendants vanish, while the genus zero invariants are given by the residue formula for deformed superpotential $W = t_0 +e^{iY}+ q e^{-iY}e^{t_1}$. The nontrivial structure of the deformation, the $e^{t_1}$ factor, originates from the K. Saito theory of good sections \cite{saito1983period}. Hence, we will refer to an LG model with the good section as the Landau-Ginzburg-Saito (LGS) to indicate the importance of good section data.

Losev \cite{Losev:1992tt}, Eguchi \cite{eguchi1993topological}, and Losev-Polyubin \cite{losev1995connection} defined gravitational descendants in LGS theory with polynomial superpotentials. The follow-up works \cite{eguchi1995topological, Eguchi:1996tg} extended Losev's construction to the mirror superpotential for $\mathbb{P}^1$, but their approach was limited only to a single descendant observable. Givental \cite{givental2001gromov} formulated the modern realization of the construction as a period integral formula for a single descendant GW invariant on toric varieties. 

Takahashi \cite{takahashi1998primitive} showed that the LGS theory defines (a tree-level) cohomological quantum field theory so that we can define correlation functions of gravitational descendants at genus zero. Moreover, he used the Kontsevich-Manin map \cite{kontsevich1998relations} for gravitational descendants to formulate and prove the mirror relation for $\mathbb{P}^1$ with gravitational descendants at genus zero. Takahashi's construction is more of an existence theorem than an explicit way to evaluate the GW invariants on $\mathbb{P}^1$ using LGS theory.

Almost 20 years later, Norbury-Scott \cite{norbury2014gromov} conjectured a topological recursion approach for the descendant GW invariants on $\mathbb{P}^1$, proven in \cite{dunin2014identification}. The curve and the holomorphic 1-form in topological recursion are identical to the LGS model, so we can classify it as a B-model type of computation. The topological recursion is an explicit way of doing the B-model computations, but it does not have the expected LG model residue type expressions. Moreover, the topological recursion approach does not generalize to the higher-dimensional toric varieties, while the LGS does!

We generalize Losev's construction for the correlation functions of gravitational descendants in the LGS theory to include mirror superpotentials. Our definition is recursive. The $n$-point correlation function vanishes if the total descendant level exceeds $n-3$. The residue formula gives the correlation function with the total descendant level equal to $n-3$. If the total descendant level is less than $n-3$, then we have at least one level zero descendant, and we use the LGS recursion to reduce the number of observables by one while keeping the same descendant level. The multiple applications of LGS recursion turn an arbitrary correlation function into a linear combination of extreme ones. 

Our definition of the LGS correlation functions with descendants obeys the cohomological quantum field theory's puncture, divisor, and topological recursion relations. Our main theorem is that the GW invariant with descendants on $\mathbb{P}^1$ at genus zero equals the LGS correlation function of the mirrored descendant observables. The mirror map for observables is the Kontsevich-Manin \cite{kontsevich1998relations} mirror map. Our proof is based on Dubrovins's reconstruction theorem for the GW invariants of gravitational descendants. Namely, we show the LGS divisor, puncture, and TRR relations map to the corresponding relations in GW theory under the Kontsevich-Manin mirror map. 

We provide several examples of the GW invariants evaluated using the mirror LGS theory. We use the LGS mirror description to get a new proof of the polynomiality of the GW invariants formulated by Norbury-Scott \cite{norbury2014gromov}. We also prove an integrality property for the descendant GW invariants.

The structure of the paper is as follows. Section 1 reviews the properties of GW invariants for $\mathbb{P}^1$. Section 2 defines the LGS correlation functions for descendants and discusses the corresponding puncture dilaton and topological recursion relations. In section 3, we formulate the mirror map relation between the GW descendants and LGS descendants, formulate and prove our main theorem on the mirror relation between the GW invariants and LGS correlation functions with descendants. In section 4, we demonstrate the mirror theorem for several well-known examples of the GW invariants. In the last section, we used the mirror relation to describe the Hurwitz numbers and prove polynomiality and integrality properties of the descendant GW invariants.

\section{Gromov-Witten invariants}
The GW descendant invariants for $X$
\be\label{eq_GW_definition}
\< \tau_{m_1} (\g_1)\cdots \tau_{m_n}(\g_n)\>^X_{g,\beta} = \int_{\overline{ \mathcal{M}_{g,n}(X,\beta)}} \bigwedge_{\a=1}^n \;\;  \psi_\a^{m_\a}\;\;ev_\a^\ast \g_\a,
\ee
where $\beta \in H_2(X)$ is the degree and $\gamma_\a \in H_{\ast}(X)$ are (smooth) represenatives of cohomology classes on $X$. The GW invariant (\ref{eq_GW_definition}) for zero-level descendants counts the holomorphic curves of genus $g$ and degree $d$ passing through cycles $\g_k$ on a target space $X$. The GW descendant invariants do not have a simple enumerative meaning. 

In our discussion, we will restrict our attention to the $g=0$ GW invariants for $X = \mathbb{P}^1$, so we will suppress the upper script $X$ and subscript $g$ in (\ref{eq_GW_definition}) The $H^{\ast}(\mathbb{P}^1)$ is generated by a trivial class $\g_I=1$ and a point class $\g_P =\omega$. The degree is labeled by a single integer $d\geq 0$. For convenience, we introduce a generating function for genus-zero invariants
\be
\begin{split}
\< \tau_{m_1} (\g_1)\cdots \tau_{m_n}(\g_n)\>  &= \sum_{d=0}^\infty q^d\; \< \tau_{m_1} (\g_1)\cdots \tau_{m_n}(\g_n)\>_{d}. 
\end{split}
\ee

\subsection{Relations between GW invariants}
Below, we briefly review the standard relations for GW invariants. For more details, we recommend Manin's book \cite{manin1999frobenius}. The moduli space integral (\ref{eq_GW_definition}) is non-zero only if the degree of the form matches the degree of the moduli space.
\begin{Proposition}{\bf (Degree selection)}
The genus-zero GW invariant (\ref{eq_GW_definition}) on $\mathbb{P}^1$ vanishes unless 
\be
\dim \overline{\mathcal{M}_{0,n} (\mathbb{P}^1,d)}=2 (g-1)+d +n = \sum_{\a=1}^n \left(m_\a+\frac12 \deg \g_\a \right).
\ee
\end{Proposition}
\begin{Theorem} {\bf (Puncture equation)} For either $n\geq 3$ or $d\geq 1, n\geq 1$
\be\label{eq_GW_puncture}
\begin{split}
\< \tau_{0}(I) & \tau_{m_1}(\g_1)\cdots \tau_{m_n}(\g_n)\>_d = \sum_{i=1}^n \<\tau_{m_1}(\g_1)\cdots \tau_{m_i-1}(\g_i)   \cdots \tau_{m_n}(\g_n)\>_d.
\end{split}
\ee 
\end{Theorem}
The special case $n=2$ 
\be\label{eq_GW_3_pt_intersection}
\< \tau_0(I) \tau_0 (\gamma_1) \tau_0(\gamma_2)\>_0 = \int_{\mathbb{P}^1} \gamma_1\wedge \gamma_2.
\ee
\begin{Theorem} {\bf(Dilaton equation)} For either $n\geq 3$ or $d\geq 1, n\geq 1$ 
\be
\<\tau_{1}(I) \tau_{m_1}(\g_1)\cdots \tau_{m_n}(\g_n)\>_d = (n-2)\< \tau_{m_1}(\g_1) \cdots \tau_{m_n}(\g_n)\>_{d}.
\ee
\end{Theorem}
\begin{Theorem} {\bf (Divisor equation)} For either $n\geq 3$ or $d>1, n\geq 1$
\be\label{eq_divisor_gw}
\begin{split}
\< \tau_{0}(P) & \tau_{m_1}(\g_1)\cdots\tau_{m_n}(\g_n)\>_{d} =d \cdot \< \tau_{m_1}(\g_1)\cdots\tau_{m_n}(\g_n)\>_{d} \\
&+ \sum_{i=1}^n \<  \tau_{m_1}(\g_1)\cdots  \cdot \tau_{m_{i-1}}(\g_{i-1}) \cdot \tau_{m_i-1} (\g_i \wedge \g_P) \cdot  \tau_{m_{i+1}} ( \g_{i+1})\cdots \tau_{m_n} (\g_n)\>_{d}.
\end{split}
\ee 
\end{Theorem}
The special case is
\be
\< \tau_0(P) \tau_0(\gamma_1) \tau_0 (\g_2)\>_{0} = \int_{\mathbb{P}^1} \omega\wedge \gamma_1\wedge \g_2.
\ee

\subsection{Toplogical recursion relation}
We define an intersection matrix $g_{ab}$ on $H^\ast (\mathbb{P}^1)$ 
\be
g_{ab} = \int_{\mathbb{P}^1} \g_a\wedge \g_b,\;\; g_{IP} = g_{PI} = 1,\;\; g_{II} = g_{PP}=0.
\ee
We denote by $g^{ab}$ the inverse matrix for $g_{ab}$. 
\begin{Theorem} {\bf (Topological recursion relation, Witten \cite{witten1990two})} For $n\geq 3$ and $d>0$
\be\label{eq_gw_trr}
\begin{split}
\< \tau_{m_1} (\gamma_1) \tau_{m_2}(\gamma_2)\cdots & \tau_{m_n}(\gamma_n) \> = \sum_{S_1\cup S_2 = \{4,\ldots,n\}}  \< \tau_0(\g_a) \tau_{m_1-1} (\gamma_1)\;  \prod_{i\in S_1} \tau_{m_i} (\gamma_i) \> \\
&\qquad \times g^{ab} \< \tau_0(\g_b) \tau_{m_2}(\gamma_2)\; \tau_{m_3}(\gamma_3) \; \prod_{j\in S_2} \tau_{m_j} (\gamma_j)\>.
\end{split}
\ee
\end{Theorem}
\begin{Theorem}\label{thm_dubrovin_reconstruction}
{\bf (Dubrovin)}. The genus-zero descendant GW invariants on $\mathbb{P}^1$ are uniquely defined by genus-zero GW invariants.
\end{Theorem}
\begin{proof} The topological recursion relation (\ref{eq_gw_trr}) allows us to decrease the total descendant level by one for each application. Hence, we can express any genus-zero descendant invariant via ordinary GW invariants. However, several issues need to be resolved. The topological recursion formula requires three or more observables and may include the 2-point GW invariants, while the ordinary GW invariants are only defined for three or more points. We can solve this problem by the divisor relation (\ref{eq_divisor_gw}) to express the 2-point GW invariants via the three-point invariants by adding a point observable. Namely,
\be\label{eq_2_pt_divisor_GW}
\begin{split}
 d \< \tau_{m_1}(\g_1) \tau_{m_2}(\g_2)\>_d  &= \< \tau_{m_1}(\g_1)\tau_{m_2}(\g_2)P\>_d \\
 &\qquad - \< \tau_{m_1-1}(\g_1\wedge \g_P)\tau_{m_2}(\g_2)\>_d - \< \tau_{m_1}(\g_1)\tau_{m_2-1}(\g_2\wedge \g_P)\>_d.
 \end{split}
\ee
If $\g_{1}=\g_{2}=\g_{P} $, then the 2-point functions in the second line of (\ref{eq_2_pt_divisor_GW}) vanish. If one of $\g_1, \g_2$ is the identity observable, then we use the (\ref{eq_2_pt_divisor_GW}) one more time for the 2-point functions in the second line. If both $\g_{1}, \g_2$ are identity observables, we need to use the (\ref{eq_2_pt_divisor_GW}) one more time. Hence, after at most three applications of the $(\ref{eq_2_pt_divisor_GW})$, we will get a 3-point function representation of the 2-point GW invariant $\< \tau_{m_1}(\g_1) \tau_{m_2}(\g_2)\>_d$.
\end{proof}

\subsection{Selected GW invariants}
For the no descendant case, the only non-trivial invariants are in degrees zero and one. The only non-vanishing degree-zero invariant
\be
  \< \tau_{0}(I) \tau_{0}(I) \tau_{0}(P)\>_0 = \< II P\>_0 =1.
\ee
The non-vanishing genus-zero invariants are degree-1, i.e.
\be\label{eq_GW_n_point_points}
  \< \underbrace{\tau_{0}(P)\cdots\tau_{0}(P)}_{n}\> =\<\underbrace{PP\cdots P}_{n}\> =q,\;\; \;\; n=1,2,3,4,\ldots
\ee
Indeed, all such invariants are generated by the divisor relation 
\be
  \<\underbrace{PP\cdots P}_{n}\>_d = d \cdot \<\underbrace{PP\cdots P}_{n-1}\>_d.
\ee
The 1-point functions for descendant invariants 
\be
  \<\tau_{2d-2}(P)\> = \frac{q^d}{(d!)^2},\;\; \<\tau_{2d-1}(I)\> = -2\frac{q^d }{(d!)^2} \sum_{k=1}^d \frac{1}{k}.
\ee
It is convenient to introduce the harmonic numbers 
\be
H_n = \sum_{k=1}^n \frac{1}{k},
\ee
and the following numbers
\be\label{eq_3_pt_coefficients}
\a_a = \left\{ \begin{array}{cc}
\frac{1}{k!^2}, & a = 2k;  \\
\frac{1}{k! (k+1)!}, &a = 2k+1.
\end{array}
\right.,\;\;
\b_a = \left\{ \begin{array}{cc}
-\frac{2k H_k-1}{k!^2} , & a = 2k;  \\
-\frac{2}{k!^2} H_k, &a = 2k+1.
\end{array}
\right. .
\ee
The selected two-point functions
\be\label{eq_2_pt_functions_GW}
\begin{split}
 \<\tau_{2k-1}(P)I\> & =\a_{2k}\;q^k,\;\;\; \<\tau_{2k}(P)P\>  = \a_{2k+1}\;q^{k+1},\\
 \<\tau_{2k-1}(I)P\>  & = \b_{2k}\; q^k,\;\;\<\tau_{2k}(I)I\>   = \b_{2k+1}\; q^{k+1}.
 \end{split}
\ee
The 3-point functions are essentially products of numbers (\ref{eq_3_pt_coefficients}). For example, 
\be
  \begin{split}
    \<\tau_{a} (P) \tau_b(P) \tau_c(P)\> &= \a_a \a_b \a_c \;\;q^{\frac12 (a+b+c+2)},\\
    \<\tau_{a} (I) \tau_b(I) \tau_c(I)\> &= \b_a \b_b \b_c \;\;q^{\frac12 (a+b+c+1)}.
  \end{split}
\ee

\section{Landau-Ginzburg-Saito theory}
In this section, we define the LGS theory for descendant invariants, compare it with the matter representation approach, and prove the puncture, dilaton, and divisor relations. We formulate and prove the topological recursion relation for the LGS theory.  

\subsection{Definitions and notations}
\begin{Definition} The LGS data ($X$, $\Omega$, $W$, $S_W$) is a collection of the following data
  \begin{itemize}
    \item Complex manifold $X$ with an algebra $\cO(X)$ of holomorphic functions on $X$;
    \item The holomorphic top form $\Omega $ on $X$;
    \item The {\it superpotential} $W$ is a holomorphic function on $X$ with isolated critical points. The $I_{dW}$ is the gradient ideal for $W$ and $J_W = \cO(X)/ I_{dW}$ is the Jacobi ring associated to $W$. The canonical projection $\pi_W: \cO(X) \to J_W$.
    \item K. Saito's {\it good section} $S_W: J_W \to \cO(X)$. The image of a good section is spanned by $\vp_{\a}\;\;\a  = 1, \ldots, \dim J_W$, the images of classes from the Jacobi ring.
  \end{itemize}
\end{Definition}
In our definition, we use the concept of good section from work \cite{saito1983period} of K. Saito, formulated in the language of holomorphic functions by Losev \cite{Losev:1994whn, Losev:1998dv}.  
\begin{Example}\label{ex_lgs_pure_grav}
  The simplest example of the LGS theory has a complex manifold $X = \mathbb{C}$ with a holomorphic coordinate $x$. The holomorphic functions are polynomials, so $\cO(X) = \mathbb{C}[x]$. The holomorphic top form and superpotential are $\Omega = dx$ and $W = \frac12 x^2$. The Jacobi ring is one-dimensional and generated by an identity class $[I]$, equipped with a trivial good section $S_W([I]) = \vp_{I} =  1 \in \mathbb{C}[x]$.
\end{Example}

\subsection{Mirror LGS theory}
The LGS theory, in example \ref{ex_lgs_pure_grav}, is the {\it mirror LGS theory} to the GW theory of a point. Using tropical geometry and topological quantum mechanics, the authors of \cite{Losev:2022tzr,losev2023tropical} constructed mirror LGS models for the GW theory on toric varieties.

Our main object of study is the mirror LGS theory for the GW invariants of a projective line $\mathbb{P}^1$. The LGS data is
\begin{itemize}
  \item The complex manifold $X = \mathbb{C}^\ast$ with holomorphic coordinate $x$. The holomorphic functions are $\cO(X) = \mathbb{C}[x, x^{-1}]$.
  \item The holomorphic top form $\Omega = \frac{dx}{x}$.
  \item The superpotential
  \be
    W = x + \frac{q}{x}.
  \ee
  The Jacobi ring is two-dimensional $J_W \simeq \mathbb{C}^2$, generated by an identity $I$ and point $P$ classes. 
  \item The good section $S_W: J_W \to \cO(\mathbb{C}^\ast)$ is defined by its image on two classes
\be
  S_W(I)= \vp_I = 1,\;\;\; S_W(P) = \vp_P = \frac{q}{x}.
\ee
\end{itemize}
Using a different holomorphic coordinate $Y$, related to the original coordinate by $x = e^{iY}$, is more convenient. The holomorphic top form is $\Omega  = i dY$, while the superpotential is 
\be\label{eq_mirr_superpot_y_coord}
W = e^{iY} + q e^{-iY}.
\ee
We also introduce a simplified notation for the residue integral of a function $F(Y)$
\be
\oint \frac{F}{W'}  = \int_{0}^{2\pi} \frac{dY}{2\pi} \frac{F}{W'}.
\ee

\subsection{Correlation functions for LGS descendants}
In works \cite{Losev:1992tt,losev1995connection}, Losev and Polyubin formulated the recursive formula for the LGS correlation functions. Here, we generalize their formula to include the LGS descendants.
\begin{Definition} For K. Saito's good sections $\varphi_\a,\; \a=1,\ldots,n>2$ and collection of non-negative integers $m_k\geq 0$ the {\it $n$-point descendant correlation function} $\<z^{m_1}\varphi_1,\ldots,z^{m_n}\varphi_n\>_{W}$ in LGS theory with superpotential $W$ is defined recursively:
\begin{itemize}
\item Over-extreme correlation function for $\sum m_k > n-3$ vanishes.
\item The extreme correlation function, for $\sum m_k = n-3$
\be\label{def_extr_LGS}
\begin{split}
\<z^{m_1}\varphi_1,\ldots,z^{m_n}\varphi_n\>_{W}=\binom{n-3}{m_1,\ldots,m_n} \oint \frac{\varphi_1\varphi_2 \cdots \varphi_n}{ W'}.
\end{split}
\ee
\item Under-extreme correlation function for $\sum m_k > n-3$ has at least one level-0 descendant and is defined recursively. The recursion with respect to the level-0 descendant $\varphi_{n}$, i.e.
\be\label{n_pt_rec}
\begin{split}
\<z^{m_1}\varphi_1,&\ldots,  z^{m_{n-1}}\varphi_{n-1},\varphi_{n}\>_{W} = \frac{d}{d\e} \Big|_{\e=0}  \< z^{m_1}\varphi^\e_1,z^{m_2}\varphi^\e_2,\ldots,z^{m_{n-1}}\varphi^\e_{n-1}\>_{W^\e}.
\end{split}
\ee
The deformation superpotential and observables
\be\label{eq_contact_term_desc_expansion}
\begin{split}
W^\e &= W + \e \vp_n,\\
z^{m}\varphi^\e &=z^m \vp +\e C_W(z^m\vp, \vp_n)=z^m \vp+ z^{m-1}S_W \pi_W(\varphi \varphi_n) +  z^m C_W(\varphi, \varphi_n).
\end{split}
\ee
\end{itemize}
\end{Definition}
For the the mirror of $\mathbb{P}^1$, we have the following simplifications 
\begin{itemize}
\item The deformation of a superpotential by a good section for a point observable is identical to the rescaling of the Kahler module of $W$, i.e.
\be
 W (Y; q) + \e \varphi_P =W (Y; q) +\e   qe^{-iY} = W(Y; q(1+\e)).
\ee
\item For the mirror of $\mathbb{P}^1$ the contact terms in (\ref{eq_contact_term_desc_expansion}) simplify to 
\be
\begin{split}
C_{W}(z^{m}\vp_I, \vp_I) &= z^{m-1}\vp_I,\;\; C_{W}(z^{m}\vp_P, \vp_I) = z^{m-1} \vp_P,\\
C_{W}(z^{m}\vp_I, \vp_P) &= z^{m-1}\vp_P,\; C_{W}(z^{m}\vp_P, \vp_P) = z^{m-1}q \vp_I+  z^m \vp_P.
\end{split}
\ee
\item The residue formula for good sections evaluates at
\be\label{eq_residue_good_sections_p1}
\oint \frac{1}{W'(Y)} \vp_P^n \vp_I^k =\frac{1}{2\pi}\int_{S^1} \frac{dY}{W'(Y)}\;q^n e^{-inY} = q^{\frac{n-1}{2}} \cdot (n \mod 2).
\ee

\end{itemize}
\begin{Remark} Works \cite{Losev:1992tt,losev1995connection} of Losev and Polyubin and the work of Eguchi, Hori, and Yang
\cite{eguchi1995topological} introduced a ``matter representation" for the gravitational descendants for the polynomial superpotentials of a single variable. In particular, the matter representation for the level-$m$ LGS descendant
\be
z^m \vp \sim   W' \int W' \int \cdots  W'\int \vp \in \cO(X).
\ee
The recursive definition (\ref{n_pt_rec}) of the LGS correlation functions allows for holomorphic functions as arguments. The contact terms in the recursion (\ref{n_pt_rec}) are well-defined for the pair of holomorphic functions. However, the derivative term in (\ref{n_pt_rec}) requires a significant modification: The shift of the superpotential by a holomorphic function is not a versal deformation in general, so we need to perform a certain infinitesimal coordinate transformation to express the answer in terms of the $(n-1)$-point correlation functions of deformed superpotential. 
\end{Remark}
\begin{Example} The 4-point function of the four level-zero descendants of point  
 \be\nn
 \begin{split}
  \< \vp_P,\vp_P,\vp_P,\vp_P\>_W & = \frac{d}{d\e} \Big|_{\e=0} \< \vp_P,\vp_P,\vp_P\>_{W (Y; q(1+\e))} + 3 \< \vp_P,\vp_P,   C_W(\varphi_P,\varphi_P)\>_W \\
  &= \frac{d}{d\e} \Big|_{\e=0} \oint \frac{q^3 e^{-3iY}}{W' (Y; q(1+\e))}  + 3\< \vp_P,\vp_P,\vp_P\>_W = \frac{d}{d\e} \Big|_{\e=0} \frac{q}  {(1+\e)^2} +3q\\ 
  &= -2q +3q =q. 
 \end{split}
 \ee
\end{Example}
\begin{Example}\label{ex_4_point_extreme}
 The 4-point function of the three level-zero descendants of point and one level-one descendant of identity
 \be\nn
 \begin{split}
  \<\vp_P,\vp_P,\vp_P, z\vp_I\>_W & =\oint \frac{q^3 e^{-3iY} \cdot 1}{W' (Y; q)}  = q.
 \end{split}
 \ee
 Alternatively, we can use the LGS recursion for the $\vp_P$-observable. Namely, 
  \be\nn
 \begin{split}
  \< \vp_P,&\vp_P,\vp_P,z\vp_I\>_W  = \frac{d}{d\e} \Big|_{\e=0} \< \vp_P,\vp_P,z\vp_I\>_{W (Y; q(1+\e))} + 2 \< \vp_P,z\vp_I,   C_W(\varphi_P,\varphi_P)\>_W\\
  &\qquad+ \< \vp_P,\vp_P, C_W(\varphi_P,z\varphi_I)\>_W=0+\< \vp_P,z\vp_I,\varphi_P\>_W+\< \vp_P,\vp_P,\varphi_P\>_W =q.
 \end{split}
 \ee
 In the second and third equalities, we used the vanishing property of the over-extreme 3-point functions.
\end{Example}
\begin{Example} 
  The 5-point function of the four level-zero descendants and one level-one descendant of a point
  \be
  \begin{split}
    \< \vp_P, &\vp_P,\vp_P,\vp_P, z\varphi_P\>_W  =\frac{d}{d\e} \Big|_{\e=0} \< \vp_P,\vp_P,\vp_P, z\varphi_P\>_{W(Y; q(1+\e))}\\
    &+ 3 \< \vp_P,\vp_P, C_W(\varphi_P,\varphi_P), z \varphi_P\>_W +  \< \vp_P,\vp_P,\vp_P, C_W(\varphi_P,z\varphi_P)\>_W \\
    &=\frac{d}{d\e}\Big|_{\e=0} \oint\frac{q^4 e^{-4iY}}{W (Y; q(1+\e))'}+3\< \vp_P,\vp_P,\vp_P,z \varphi_P\>_W +\<\vp_P,\vp_P,\vp_P,q\vp_I+z\vp_P\>_W\\
    & = 0 + 4 \oint \frac{q^4 e^{-4iY}}{W'} + q\< \vp_P,\vp_P,\vp_P, \vp_I\>_W = 0 +0 = 0.
  \end{split}
  \ee
\end{Example}

\subsection{Puncture and dilaton relations}
Losev and Polyubin \cite{losev1995connection} used the matter representation for the gravitational descendants to derive the puncture and dilaton relations for the LGS correlation functions. Here, we use our definition to derive the same relations.
\begin{Proposition} {\bf (Puncture relation)} For $n\geq 3$, $m_k\geq 0$ and good sections $\vp_1,\ldots, \vp_n$
   \be\label{eq_LGS_punct_relation}
     \<z^{m_1}\vp_1,z^{m_2}\vp_2,\ldots, z^{m_{n}}\vp_{n},\vp_I\>_{W}  = \sum_{j=1}^n \<z^{m_1}\vp_1,\ldots, z^{m_j-1}\vp_j, \ldots,z^{m_{n}}\vp_{n}\>_{W}.
   \ee
\end{Proposition}
\begin{proof} We use the $\vp_{I}$ observable to perform the LGS recursion for (\ref{eq_LGS_punct_relation})
  \be\label{eq_lgs_exp_identity_obs}
  \begin{split}
    \<&z^{m_1}\vp_1,z^{m_2}\vp_2,\ldots, z^{m_{n}}\vp_{n},\vp_I\>_{W} = \frac{d}{d\e} \Big|_{\e=0}  \< z^{m_1}\vp_1,z^{m_2}\vp_2,\ldots,z^{m_{n}}\vp_{n}\>_{W+\e \vp_I}\\
    & +\< C_{W}(z^{m_1}\vp_1, \vp_I),z^{m_2}\vp_2,\ldots,z^{m_{n}}\vp_{n}\>_{W}+\cdots+\< z^{m_1}\vp_1,\ldots,C_{W}(z^{m_{n}}\vp_{n}, \vp_I)\>_{W}.
  \end{split}
  \ee
  The derivative term vanishes since the residue formula (\ref{def_extr_LGS}) and contact terms (\ref{eq_contact_term_desc_expansion}) only depend on the derivative of the superpotential $W'$. We evaluate the contact terms 
  \be\nn
  C_{W}(z^{m_k}\vp_k, \vp_I) = z^{m-1}\vp_k
  \ee
  in the third line of (\ref{eq_lgs_exp_identity_obs}) to complete the proof.
\end{proof}

\begin{Proposition} {\bf(LGS Dilaton relation)} For $n\geq 3$, $m_k\geq 0$ and good sections $\vp_1,\ldots, \vp_n$
  \be\label{eq_LGS_dilaton_relation}
    \<z^{m_1}\vp_1,\ldots, z^{m_{n}}\vp_{n},z\vp_I\>_{W}  =(n-2) \<z^{m_1}\vp_1,\ldots, z^{m_{n}}\vp_{n}\>_{W}.
  \ee
\end{Proposition}
\begin{proof}
  We prove the equality using induction on the number of observables $n$. The base of the induction is $n=3$. If $ m_1+m_2+m_3 >0$, both sides vanish since we have over-extreme LGS correlation functions. If $m_1+m_2+m_3=0$, then there are no descendants for the rhs of (\ref{eq_LGS_dilaton_relation}), and the total level of the descendants for the lhs of (\ref{eq_LGS_dilaton_relation}) is equal to one. Hence, the rhs  of (\ref{eq_LGS_dilaton_relation}) is extreme, and we evaluate both sides to prove the equality
  \be\nn
    \<\vp_1,\vp_{2}, \vp_3, z\vp_I\>_{W}  = \oint \frac{\vp_1\vp_2\vp_3}{W'} = (3-2) \cdot \<\vp_1,\vp_{2},\vp_3\>_{W} .
  \ee
  For $n>3$, the correlation functions in (\ref{eq_LGS_dilaton_relation}) are nonzero only when at least one of $m_k$ is zero. Let us assume that $m_{n} = 0$. Let us consider the LSG recursion for the $\vp_{n}$ for the following expression
  \be\nn
  \begin{split}
    \<&z^{m_1}\vp_1,\ldots, z^{m_{n-1}}\vp_{n-1},z\vp_I,\vp_{n}\>_{W}-(n-3) \<z^{m_1}\vp_1,\ldots, z^{m_{n-1}}\vp_{n-1}, \vp_n\>_{W}\\
    & = \frac{d}{d\e} \Big|_{\e=0}  \< z^{m_1}\vp_1,\ldots,z^{m_{n-1}}\vp_{n-1}, z\vp_I\>_{W+\e \vp_{n}} - (n-3)\frac{d}{d\e} \Big|_{\e=0}  \< z^{m_1}\vp_1,\ldots,z^{m_{n-1}}\vp_{n-1}\>_{W+\e \vp_{n}}\\  
    & + \<z^{m_1}\vp_1,\ldots, z^{m_{n-1}}\vp_{n-1},C_W(z\vp_I,\vp_{n})\>_{W}+\sum_{k=1}^{n-1}\<z^{m_1}\vp_1,\ldots,C_{W}(z^{m_k}\vp_k, \vp_{n}),\ldots,z^{m_{n-1}}\vp_{n-1}, z\vp_I\>_{W}\\
   &\qquad\qquad-(n-3)\sum_{k=1}^{n-1}\< z^{m_1}\vp_1,\ldots,C_{W}(z^{m_k}\vp_k, \vp_{n}),\ldots,z^{m_{n-1}}\vp_{n-1}\>_{W}\\
    & = \<z^{m_1}\vp_1,z^{m_2}\vp_2,\ldots, z^{m_{n-1}}\vp_{n-1},\vp_{n}\>_{W}.
  \end{split}
  \ee
  The terms in the second line and the sums cancel by the induction assumption. Note that the equality 
  \be\nn
    \< z^{m_1}\vp_1,\ldots,z^{m_{n-1}}\vp_{n-1}, z\vp_I\>_{W+\e \vp_{n}} = (n-3)  \< z^{m_1}\vp_1,\ldots,z^{m_{n-1}}\vp_{n-1}\>_{W+\e\vp_{n}}
  \ee
  requires $n$-point equality for the superpotential deformed by the good section. For the mirror superpotential (\ref{eq_mirr_superpot_y_coord}), the space of deformations is 2-dimensional. Each deformation is either along the trivial class $\vp_{I}$ or the point class $\vp_{P}$. The deformation along the trivial class $\vp_{n} = \vp_{I}$ preserves the LGS correlation functions. The deformation along the point class $\vp_n = \vp_P$ is equivalent to the shift of the $q$ in the mirror superpotential within an induction assumption. Hence, we arrive at equality 
  \be\label{eq_simpl_lgs_rec_dilaton_rel}
  \begin{split}
    \<&z^{m_1}\vp_1,\ldots, z^{m_{n-1}}\vp_{n-1},z\vp_I,\vp_{n}\>_{W}-(n-3) \<z^{m_1}\vp_1,\ldots, z^{m_{n-1}}\vp_{n-1}, \vp_n\>_{W}\\
    & =\<z^{m_1}\vp_1,z^{m_2}\vp_2,\ldots, z^{m_{n-1}}\vp_{n-1},\vp_{n}\>_{W}.
  \end{split}
  \ee
  The equality (\ref{eq_simpl_lgs_rec_dilaton_rel}) implies the dilaton relation (\ref{eq_LGS_dilaton_relation}) for the $(n+1)$-point correlation function.
\end{proof}

\subsection{Divisor relation}

The puncture (\ref{eq_LGS_punct_relation}) and dilaton (\ref{eq_LGS_dilaton_relation}) relations hold for an LGS theory with a generic superpotential. For the mirror LGS theory, we have an additional divisor relation.  
\begin{Proposition}{\bf(LGS Divisor relation)} For $n\geq 3$, $m_k\geq 0$ and good sections $\vp_1,\ldots, \vp_n$
  \be\label{eq_LGS_divisor}
  \begin{split}
    \<  & z^{m_1}\varphi_{1},\ldots ,z^{m_n}\varphi_{n}, \varphi_P\>_W = q \frac{d}{dq}  \<z^{m_1}\varphi_{1},\ldots ,z^{m_n}\varphi_{n}\>_W \\
    &\qquad+ \sum_{i=1}^n  \< z^{m_1} \vp_{1}, \ldots,z^{m_i-1}S_W\pi_W(\vp_{i} \vp_P),\ldots,z^{m_n}\varphi_{n}  \>_W .
  \end{split}
  \ee
\end{Proposition} 
\begin{proof}
  We use the LGS recursion for the $\vp_P$
  \be
  \begin{split}
    \< & z^{m_1}\varphi_{1},\ldots ,z^{m_n}\varphi_{n},\varphi_P\>_W=\frac{d}{d\e} \Big|_{\e=0} \< z^{m_1}\vp_{1},\ldots ,z^{m_n}\vp_{n}\>_{W+\e \vp_P}\\
    &+ \sum_{i=1}^n  \< z^{m_1} \vp_{1}, \ldots,C_W(z^{m_i}\vp_{i}, \vp_P),\ldots,z^{m_n}\varphi_{n}  \>_W = q \frac{d}{dq}  \<z^{m_1}\varphi_{1},\ldots ,z^{m_n}\varphi_{n}\>_W \\
    &\qquad+ \sum_{i=1}^n  \< z^{m_1} \vp_{1}, \ldots,C_W(z^{m_i}\vp_{i}, \vp_P) - z^{m_i}q\frac{d}{dq}\vp_{i}  ,\ldots,z^{m_n}\varphi_{n}  \>_W. 
  \end{split}
  \ee
  We use contact term formula (\ref{eq_contact_term_desc_expansion}) and the following equality for $\vp_I =1$ and $\vp_P = q e^{-iY}$
  \be\nn
    q\frac{d}{dq}\vp - C_W (\vp, \vp_P) = 0.
  \ee
\end{proof}

\subsection{Topological recursion relation}

The connection between the LGS theory and the integrals over the moduli space of complex structures manifests in the {\it topological recursion relation for the LGS correlators}. 
\be\label{eq_LGS_TRR}
\begin{split}
  \<z^{m_1}\vp_1, z^{m_2}\vp_2, z^{m_3}\vp_3,& \ldots, z^{m_n}\vp_n\>_W = \sum_{S_1\cup S_2 = \{4,\ldots,n\}, S_1\neq \emptyset} \eta^{ab}\< \vp_a, z^{m_1-1}\vp_1, \{z^{m_\a}\vp_\a\}_{\a\in S_1}\>_W \\
  &\quad\< \vp_b, z^{m_2}\vp_2, z^{m_3}\vp_3,\{z^{m_\b}\vp_\b\}_{\b\in S_2}\>_W.
\end{split}
\ee
We will consider two subcases of the (\ref{eq_LGS_TRR}) as preparational steps for the general case proof.
\begin{Proposition} 
  LGS TRR (\ref{eq_LGS_TRR}) holds for $n=4$.
\end{Proposition}
\begin{proof} 
  The $n=4$ version of the LGS TRR (\ref{eq_LGS_TRR})
  \be\label{eq_n_4_lgs_trr}
    \<z\vp_1,\vp_2,\vp_3,\vp_4\>_W = \eta^{ab}\< \vp_a,\vp_1, \vp_4\>_W \< \vp_b,\vp_2,\vp_3\>_W.
  \ee
  The LGS correlation function on the lhs of (\ref{eq_n_4_lgs_trr}) is extreme, so a residue formula gives it 
  \be\label{eq_rhs_n_4_lgs_trr}
    \<z\vp_1,\vp_2,\vp_3,\vp_4\>_W = \oint \frac{\vp_1\vp_2\vp_3\vp_4}{W'}.
  \ee
  We decompose the product of two good sections $\vp_{2} \cdot \vp_{3}$ as an expansion in good sections $\vp_c$ and an element from the gradient ideal. Namely
  \be\label{eq_prod_jacobi_ring_expansion}
    \vp_2\cdot \vp_3 = f_{23}^c \vp_c + r_{23}W'.
  \ee
  The LGS 3-point function in the rhs of (\ref{eq_n_4_lgs_trr}) evaluates into 
  \be
    \< \vp_b,\vp_2,\vp_3\>_W = \oint \frac{\vp_b\vp_2\vp_3}{W'} = f_{23}^c \oint \frac{\vp_b \vp_c}{W'} = f_{23}^c \eta_{bc}.
  \ee
  The rhs in (\ref{eq_n_4_lgs_trr}) simplifies into the expression (\ref{eq_rhs_n_4_lgs_trr}), which is identical to the lhs of the 4-point LGS TRR. Namely, 
  \be\nn
    \eta^{ab}\< \vp_a,\vp_1, \vp_4\>_W \< \vp_b,\vp_2,\vp_3\>_W =f_{23}^c \eta_{bc} \eta^{ab} \oint \frac{\vp_a\vp_1\vp_4}{W'} =f_{23}^a \oint \frac{\vp_a\vp_1\vp_4}{W'} =\oint \frac{\vp_2\vp_3\vp_1\vp_4}{W'}.
  \ee
\end{proof}
The 4-point version (\ref{eq_n_4_lgs_trr}) of the LGS TRR is the simplest case of the bigger family of extreme TRRs.
\begin{Proposition}\label{pr_extr_lgs_trr}{\bf(Extreme TRR)} 
  LGS TRR (\ref{eq_LGS_TRR}) holds for $n\geq 3, m_1>0 $ and $m_k\geq 0$, such that $\sum m_k = n-3$.
\end{Proposition}
 \begin{proof}
   The products of the correlation functions in the rhs of the TRR (\ref{eq_LGS_TRR}) have total descendant level $\sum m_k-1 = n-4$ distributed among the $n+2$ observables in two brackets. Since $\sum m_k-1 = n-4 = n+2-3-3$, each term in the sum is either a product of two extreme correlators or the over-extreme and under-extreme. Hence, the only nonzero contributions to the rhs of the LGS TRR (\ref{eq_LGS_TRR}) are products of extreme correlators. Using the relation (\ref{eq_prod_jacobi_ring_expansion}) we arrive at an equality
   \be\label{eq_extreme_lgs_trr_structure_relation}
   \begin{split}
     &\< z^{m_1}\vp_1,  z^{m_2}\vp_2, z^{m_3}\vp_3, \ldots, z^{m_n}\vp_n\>_W =\binom{n-3}{m_1,\ldots, m_n}\oint \frac{1}{W'} \prod_{k=1}^n \vp_k\\
     &\sum_{S_1\cup S_2 = \{4,\ldots,n\}} \eta^{ab}\< \vp_a, z^{m_1-1}\vp_1, \{z^{m_\a}\vp_\a\}_{\a\in S_1}\>_W \< \vp_b, z^{m_2}\vp_2, z^{m_3}\vp_3,\{z^{m_\b}\vp_\b\}_{\b\in S_2}\>_W \\
     &= \sum_{S_1\cup S_2 = \{4,\ldots,n\}} C_{S_1S_2}\oint \frac{\vp_a\vp_1}{W'} \prod_{\a \in S_1} \vp_\a\oint \frac{\vp_b \vp_2\vp_3}{W'} \prod_{\b\in S_2} \vp_\b = C\cdot  \oint \frac{1}{W'} \prod_{k=1}^n \varphi_k.
   \end{split}
   \ee
  Hence, the proof of LGS TRR in the extreme case is equivalent to the proof of the numerical relation between the coefficients in front of the extreme correlators in (\ref{eq_extreme_lgs_trr_structure_relation}). Moreover, the coefficients are independent of the superpotential. We use the mirror superpotential for the GW theory with superpotential $W=x^2$. The LGS correlators are identical to the moduli space integrals, i.e. 
   \be\label{eq_multinomial_moduli_space_integral}
     \<z^{m_1}\vp_I,  z^{m_2}\vp_I, z^{m_3}\vp_I, \ldots, z^{m_n}\vp_I\>_{W=x^2} = \int_{\overline{M_{0,n}}} \psi_1^{m_1}\cdots \psi_n^{m_n}=\binom{n-3}{m_1,\ldots, m_n}.
   \ee
 \end{proof}
\begin{Theorem} {\bf (LGS TRR)} 
  LGS TRR (\ref{eq_LGS_TRR}) holds for $n\geq 4, m_1>0$ and $m_k\geq 0$ and a collection of good sections $\vp_1,\ldots, \vp_n$.
\end{Theorem}
\begin{proof} We prove the LGS TRR by induction in $n$. In an over-extreme case, $\sum m_k > n-3$, the lhs of the LGS TRR is zero, while the RHS is a sum of products of factors with at least one of the factors being over-extreme. We can verify this claim by counting the total descendant level of each pair of factors. 

We have already proven the extreme case, $\sum m_k = n-3$, in proposition (\ref{pr_extr_lgs_trr}). Hence, we need to prove the relation in the under-extreme cases. The under-extremality implies that there is at least one observable with $m=0$. Let us assume that $m_n=0$, use the LGS recursion for the observable $\vp_n$, and then use the LGS TRR for the $(n-1)$ points. The nature of LGS recursion is such that $z^{m_1}\vp_1^\ve$ might contain lower powers of $z$, while the LGS TRR is only valid for non-zero $z$-powers. In the case of the mirror superpotential (\ref{eq_mirr_superpot_y_coord}) for $\mathbb{P}^1$, the deformation only contains the $z^{m_1-1}$ term. Namely, 
\be\label{eq_deform_lgs_descend}
  z^{m_1}\vp_1^\ve=z^{m_1}\vp_1+\ve z^{m_1}C_W(\vp_1,\vp_n)+\ve z^{m_1-1} S_W\pi_W(\vp_1\vp_n).
\ee
 Hence, we need a separate treatment for the $S_W\pi_W(\vp_1\vp_n)$-term in the $m_1=1$ special case. For $m_1>1$, we have an equality
\be\nn
\begin{split}
  &\<z^{m_1}\vp_1, z^{m_2}\vp_2, z^{m_3}\vp_3, \ldots, z^{m_{n-1}}\vp_{n-1},\vp_n\>_W  = \frac{d}{d\ve}\Big|_{\ve=0}
  \<z^{m_1}\vp^\ve_1, \ldots, z^{m_{n-1}}\vp^\ve_{n-1}\>_{W^\ve}\\
  & =\frac{d}{d\ve}\Big|_{\ve=0} \sum_{S_1\neq \emptyset, S_1\cup S_2 =  \{4,\ldots,n-1\}} \eta^{ab}\< \vp_a, z^{m_1-1}\vp^\ve_1, \{z^{m_\a}\vp^\ve_\a\}_{\a\in S_1}\>_{W^\ve} \< \vp_b, z^{m_2}\vp^\ve_2, z^{m_3}\vp^\ve_3,\{z^{m_\b}\vp^\ve_\b\}_{\b\in S_2}\>_{W^\ve}\\
  & =\sum \eta^{ab} \frac{d}{d\ve}\Big|_{\ve=0}\< \vp_a, z^{m_1-1}\vp^\ve_1, \{z^{m_\a}\vp^\ve_\a\}_{\a\in S_1}\>_{W^\ve} \< \vp_b, z^{m_2}\vp^\ve_2, z^{m_3}\vp^\ve_3,\{z^{m_\b}\vp^\ve_\b\}_{\b\in S_2}\>_W\\
  &+\qquad \sum \eta^{ab} \< \vp_a, z^{m_1-1}\vp^\ve_1, \{z^{m_\a}\vp^\ve_\a\}_{\a\in S_1}\>_W \frac{d}{d\ve}\Big|_{\ve=0}\< \vp_b, z^{m_2}\vp^\ve_2, z^{m_3}\vp^\ve_3,\{z^{m_\b}\vp^\ve_\b\}_{\b\in S_2}\>_{W^\ve}.
\end{split}  
\ee
We dropped the arguments in the later summations since they are identical to the first sum. The $\vp_n$ may enter the LGS TRR in three possible ways: it can be part of $S_1$, part of $S_2$, or a third observable in $\<\vp_a,z^{m_1-1}\vp_1,\vp_n\>_W$. Namely, the LGS TRR is an equality
\be\label{eq_lgs_trr_lgs_recursion}
\begin{split}
  &\<z^{m_1}\vp_1, z^{m_2}\vp_2, z^{m_3}\vp_3, \ldots, z^{m_{n-1}}\vp_{n-1}, \vp_n\>_W \\
  &=\eta^{ab}\< \vp_a, z^{m_1-1}\vp_1, \vp_n\>_W\<\vp_b, z^{m_2}\vp_2, z^{m_3}\vp_3, \ldots, z^{m_{n-1}}\vp_{n-1}\>_W\\
  &+ \sum_{S_1\cup S_2 = \{4,\ldots,{n-1}\}, S_1\neq \emptyset} \eta^{ab}\< \vp_a, z^{m_1-1}\vp_1, \{z^{m_\a}\vp_\a\}_{\a\in S_1}\>_W 
  \< \vp_b, z^{m_2}\vp_2, z^{m_3}\vp_3,\{z^{m_\b}\vp_\b\}_{\b\in S_2}, \vp_n\>_W\\
  &+\sum_{S_1\cup S_2 = \{4,\ldots,{n-1}\}, S_1\neq \emptyset} \eta^{ab}\< \vp_a, z^{m_1-1}\vp_1, \{z^{m_\a}\vp_\a\}_{\a\in S_1}, \vp_n\>_W \< \vp_b, z^{m_2}\vp_2, z^{m_3}\vp_3,\{z^{m_\b}\vp_\b\}_{\b\in S_2}\>_W\\
  & =\sum \eta^{ab} \frac{d}{d\ve}\Big|_{\ve=0}\< \vp_a, z^{m_1-1}\vp^\ve_1, \{z^{m_\a}\vp^\ve_\a\}_{\a\in S_1}\>_{W^\ve} \< \vp_b, z^{m_2}\vp^\ve_2, z^{m_3}\vp^\ve_3,\{z^{m_\b}\vp^\ve_\b\}_{\b\in S_2}\>_W\\
  &+\qquad \sum \eta^{ab} \< \vp_a, z^{m_1-1}\vp^\ve_1, \{z^{m_\a}\vp^\ve_\a\}_{\a\in S_1}\>_W \frac{d}{d\ve}\Big|_{\ve=0}\< \vp_b, z^{m_2}\vp^\ve_2, z^{m_3}\vp^\ve_3,\{z^{m_\b}\vp^\ve_\b\}_{\b\in S_2}\>_{W^\ve}.
\end{split}
\ee
The first term in the second line of (\ref{eq_lgs_trr_lgs_recursion}) vanishes for $m_1>1$, since $\< \vp_a, z^{m_1-1}\vp_1, \vp_n\>_W =0$ and the two expressions are the same. Hence, given LGS TRR for $(n-1)$ observables and $m_1>0$, the LGS recursion implies LGS TRR for $n$-points. 

For $m_1=1$, the first term in (\ref{eq_lgs_trr_lgs_recursion}) simplifies
into 
\be
\begin{split}
&\eta^{ab}\< \vp_a, z^{m_1-1}\vp_1, \vp_n\>_W\<\vp_b, z^{m_2}\vp_2, z^{m_3}\vp_3, \ldots, z^{m_{n-1}}\vp_{n-1}\>_W \\
&= \<S_W\pi_W(\vp_1\vp_n), z^{m_2}\vp_2, z^{m_3}\vp_3, \ldots, z^{m_{n-1}}\vp_{n-1}\>_W.
\end{split}
\ee
The $S_W\pi_W(\vp_1\vp_n)$-term is the same term as $z^0$-term of (\ref{eq_deform_lgs_descend}) for $m_1=1$. The remaining terms in (\ref{eq_deform_lgs_descend}) have positive powers of $z$, so the LGS relation proof is the same as in the $m_1>1$ case. Hence, the LGS TRR for $n-1$ points and LGS recursion imply LGS TRR for $n$ points.
\end{proof}

\section{Mirror with descendants}

In our works \cite{Losev:2022tzr,losev2023tropical} on tropical mirror symmetry, we proved the mirror relation between the GW theory on toric space $X$ and the correlators in mirror LGS theory. We can summarize the relation in the form
\begin{Theorem} 
  For $n+k\geq 3$ 
  \be
     \<\underbrace{I\cdots I}_k\cdot \underbrace{P\cdots P}_n\>^{\mathbb{P}^1} = \< \underbrace{\vp_I,\ldots,\vp_I}_k,\underbrace{\vp_P,\ldots, \vp_P}_n\>_W.
  \ee
\end{Theorem}
\begin{proof}
The proof uses the explicit evaluation of both sides. The GW puncture relation (\ref{eq_GW_puncture}) implies that the invariant vanishes for $k>0$ and $n\geq 3$. The LGS puncture relation (\ref{eq_LGS_punct_relation}) implies that the LGS correlator also vanishes for $k>0$ and $n\geq 3$. We use the divisor relation to evaluate the GW invariant for the $k=0$ case (\ref{eq_GW_n_point_points}). The LGS divisor relation (\ref{eq_LGS_divisor}) implies that
\be\label{eq_n_point_lgs_points}
\<\underbrace{\vp_P,\ldots, \vp_P}_n\>_W = (q\p_q)^{n-3} \< \vp_P,\vp_P, \vp_P\>_W = (q\p_q)^{n-3} \oint \frac{\vp_P^3}{W'} = (q\p_q)^{n-3} q = q.
\ee
The final check is the $k=2, n=1$ correlation function. The GW expression (\ref{eq_GW_3_pt_intersection}) matches the LGS expression via 
\be\nn
\<IIP\>^{\mathbb{P}^1} =1 = \<\vp_I\vp_I\vp_P\> = \oint \frac{\vp_I^2\vp_P}{W'} =1.
\ee
\end{proof}
The proof of the theorem by an explicit evaluation shows that the mirror symmetry for $\mathbb{P}^1$ is extremely simple. However, in the presence of gravitational descendants, the GW invariants are more complicated, and the proof of the mirror symmetry is far from trivial.

\subsection{Mirror map for descendants}
Using coefficients (\ref{eq_3_pt_coefficients}) for $k\geq 0$ we introduce the following elements of the good section
\be\label{eq_KM_coeffcients}
\begin{split}
C_{2k}(P) &= \a_{2k} q^k \vp_P  = \frac{q^k}{k!^2} \vp_P,\;\;\; C_{2k+1}(P) = \a_{2k+1} q^{k+1}\vp_I =\frac{q^{k+1}}{k! (k+1)!}\vp_I;\\
C_{2k}(I) &= \b_{2k} q^k \vp_I = \frac{q^k}{k!^2}(1-2k H_k)\vp_I,\;\; C_{2k+1}(I) = \b_{2k+1} q^{k}\vp_P =-\frac{2q^k}{k!^2}H_k \vp_P .
\end{split}
\ee
\begin{Definition}{\bf (Kontsevich-Manin map)} For $m\geq 0$ and cycle $\g$ the {\it mirror descendant} is the LGS observable
\be\label{eq_KM_map}
\Phi_{m}(\g) = \sum_{k=0}^m z^{k}C_{m-k}(\g) = C_m(\g) +z \Phi_{m-1}(\g).
\ee
\end{Definition}
Kontsevich and Manin in \cite{kontsevich1998relations} instead of (\ref{eq_KM_map}) used a different formula
\be\label{eq_KM_map_2pt_funct}
\Phi_{m}(\g) =\< \tau_{m-1}(\g) \g_a\> g^{ab} \vp_b +z \Phi_{m-1}(\g),\;\; \Phi_0(\g) = \vp_\g.
\ee
However, the two expressions (\ref{eq_KM_map}) and (\ref{eq_KM_map_2pt_funct}) become the same if we use the 2-point functions (\ref{eq_2_pt_functions_GW}) for the GW theory on $\mathbb{P}^1$.

The first several descendants of a point class are
\be\label{eq_low_point_desc_mirr_map}
\begin{split}
\Phi_0(P) &= C_0(P) = \vp_P = qe^{-iY},\\
\Phi_{1}(P) &= C_1(P) + z \Phi_{0}(P) =  q \vp_I+z\vp_P = q + z qe^{-iY},\\
\Phi_{2}(P) &= C_2(P) + z \Phi_{1}(P) =  q\vp_P+zq\vp_I+z^2 \vp_P,\\
\Phi_{3}(P) &= C_3(P) + z \Phi_{2}(P) = \frac12 q^2\vp_I +zq\vp_P+z^2q+z^3 \vp_P.
\end{split}
\ee
The first several descendants of the identity class are
\be\label{eq_low_identity_desc_mirr_map}
\begin{split}
\Phi_0(I) &=C_0(I) =  \vp_I =1,\\
\Phi_1(I) &=C_1(I) + z \Phi_{0}(I)  =  z\vp_I =z,\\
\Phi_2(I) &= C_2(I) + z \Phi_{1}(I)  = -q\vp_I +z^2 \vp_I = -q+z^2,\\
\Phi_3(I) &= C_3(I) + z \Phi_{2}(I) = -2q\vp_P -qz\vp_I +z^3\vp_I = -2q^2 e^{-iY} -qz+z^3.
\end{split}
\ee

\subsection{Mirror for correlation functions}

Our main result is the equality between the LGS correlation functions and GW invariants.
\begin{Theorem} \label{thm_main}
  The correlation functions in the descendant GW theory on $\mathbb{P}^1$ are identical to the LGS correlation functions in the mirror LGS theory of the corresponding mirror LGS descendant invariants. Namely, for $n\geq 3$, $m_n\geq 0$ and $\g_n \in H^\ast (\mathbb{P}^1)$
  \be
     \<\tau_{m_1}(\g_1)\cdots \tau_{m_n}(\g_n)\> = \< \Phi_{m_1}(\g_1),\ldots, \Phi_{m_n}(\g_n)\>_W.
  \ee
\end{Theorem}
\begin{proof} 
  The Dubrovin reconstruction theorem \ref{thm_dubrovin_reconstruction} implies that the correlation functions of descendants in GW theory at genus zero are uniquely restored from the TRR. We must also include the GW theory without descendants, puncture, and divisor relations. We showed the match between the GW invariants on $\mathbb{P}^1$ and the LGS correlation functions of the corresponding good sections. Hence, we only need to show that the LGS correlation functions of the mirror observables satisfy the GW puncture, divisor, and TRR relations. 
\end{proof}

\subsection{Puncture and divisor relations}

\begin{Proposition}
  The mirror map and the LGS puncture relation imply the GW puncture relation. Namely 
  \be
     \< \Phi_{m_1}(\g_1),\ldots,\Phi_{m_n}(\g_n),\vp_I\>_W = \sum_{k=1}^n \< \Phi_{m_1}(\g_1),\ldots,\Phi_{m_k-1}(\gamma_k), \ldots,\Phi_{m_n}(\g_n)\>_W.
  \ee
\end{Proposition}
\begin{proof}
  For $n\geq 3$, we use the LGS recursion with respect to the deformation by an identity observable 
  \be
      \< \Phi_{m_1}(\g_1),\ldots,\Phi_{m_n}(\g_n),\vp_I\>_W=\sum_{k=1}^n \<\Phi_{m_1}(\g_1),\ldots,C_W(\Phi_{m_k}(\gamma_k),\vp_I),\ldots,\Phi_{m_n}(\g_n)\>_W.
  \ee
  The contact terms evaluate into
  \be
     C_W(\Phi_{m}(\gamma),\vp_I) = C_W(C_m(\gamma) +z\Phi_{m-1}(\gamma),\vp_I) = \Phi_{m-1}(\gamma).
  \ee
\end{proof}
\begin{Proposition}
   The mirror map and the LGS divisor relation imply the GW divisor relation. Namely, for $n\geq 3$
   \be
   \begin{split}
      \< &\Phi_{m_1}(\g_1),\ldots, \Phi_{m_n}(\g_n), \varphi_P\>_W=  q \frac{d}{dq}\< \Phi_{m_1}(\g_1),\ldots, \Phi_{m_n}(\g_n)\>_W \\
      &+ \sum_{k=1}^n \< \Phi_{m_1}(\g_1),\ldots,\Phi_{m_k-1}(\g_k \wedge \g_P) \ldots,\Phi_{m_n}(\g_n)\>_W.
   \end{split}
   \ee
\end{Proposition}
\begin{proof}
  We use the LGS recursion with respect to the deformation by $\vp_P$
  \be\label{eq_proof_lgs_div_lgs_rec}
  \begin{split}
    \<& \Phi_{m_1}(\g_1),\ldots, \Phi_{m_n}(\g_n), \vp_P\>_W  = \frac{d}{d\ve}\Big|_{\ve =0} \< \Phi_{m_1}(\g_1),\ldots, \Phi_{m_n}(\g_n)\>_{W+\ve \vp_P}  \\
    &+\sum_{k=1}^n \< \Phi_{m_1}(\g_1),\ldots,C_W(\vp_P, \Phi_{m_k}(\gamma_k)) , \ldots,\Phi_{m_n}(\g_n)\>_W.
  \end{split}
  \ee
  We rewrite the derivative term 
  \be\label{eq_proof_lgs_div_der_terms}
  \begin{split}
    \frac{d}{d\ve}\Big|_{\ve =0} \< \Phi_{m_1}(\g_1),\ldots,& \Phi_{m_n}(\g_n)\>_{W+\ve \vp_P} = q \frac{d}{dq}\< \Phi_{m_1}(\g_1),\ldots, \Phi_{m_n}(\g_n)\>_W\\
    & - \sum_{k=1}^n \< \Phi_{m_1}(\g_1),\ldots, q \frac{d}{dq}\Phi_{m_k}(\g_k), \ldots,\Phi_{m_n}(\g_n)\>_W.
  \end{split}
  \ee
  We combine the contact terms for descendants of $P$ in (\ref{eq_proof_lgs_div_lgs_rec}) and derivative terms in (\ref{eq_proof_lgs_div_der_terms}) and rewrite 
  \be\label{eq_proof_lgs_div_combination}
  \begin{split}
    &C_W(\vp_P, \Phi_{m}(\g))-q \frac{d}{dq}\Phi_{m}(\g) = zC_W(\vp_P,\Phi_{m-1}(\g)) -zq \frac{d}{dq}\Phi_{m-1}(\g)\\
    &+C_W(\vp_P, C_m(\g)) -q \frac{d}{dq}C_m(\g) +S_W\pi_W(\vp_P\cdot C_{m-1}(\g)) .
  \end{split}
  \ee
  The combination (\ref{eq_proof_lgs_div_combination}) of contact terms and derivative terms is an inductive relation, so we only need to verify that the second line of (\ref{eq_proof_lgs_div_combination}) matches with $C_{m-1}(\g \wedge \g_P)$. 
  
  For even descendants of $P$, we evaluate
  \be
  \begin{split}
     &C_W(\vp_P, C_{2k}(P)) -q \frac{d}{dq}C_{2k}(P) +S_W\pi_W(\vp_P\cdot C_{2k-1}(P)) \\
     &= C_W\left(\vp_P,  \frac{q^k}{k!^2} \vp_P \right) -q \frac{d}{dq}\left(\frac{q^k}{k!^2} \vp_P\right) +S_W\pi_W\left(\vp_P\cdot \frac{q^k}{k!(k-1)!} \vp_I\right)\\
     &= \frac{q^k}{k!^2} \vp_P -(k+1)\frac{q^k}{k!^2} \vp_P +\vp_P \frac{q^k}{k!(k-1)!} = 0.
  \end{split}
  \ee
  Note that there is an extra power of $q$ in $\vp_{P} = q e^{-iY}$. For odd descendants of $P$, we evaluate
  \be
  \begin{split}
     &C_W(\vp_P, C_{2k+1}(P)) -q \frac{d}{dq}C_{2k+1}(P) +S_W\pi_W(\vp_P\cdot C_{2k}(P)) \\
     &= C_W\left(\vp_P,  \frac{q^{k+1}}{k!(k+1)!} \vp_I \right) -q \frac{d}{dq}\left(\frac{q^{k+1}}{k!(k+1)!}\vp_I\right) +S_W\pi_W\left(\vp_P\cdot \frac{q^k}{k!^2}\vp_P\right)\\
     &= 0 -(k+1)\frac{q^{k+1}}{k!(k+1)!} \vp_I+\frac{q^k}{k!^2}q \vp_I = 0.
  \end{split}
  \ee
For even descendants of $I$, we evaluate
\be
\begin{split}
&C_W(\vp_P, C_{2k}(I)) -q \frac{d}{dq}C_{2k}(I) +S_W\pi_W(\vp_P\cdot C_{2k-1}(I)) \\
&= C_W\left(\vp_P,  C \cdot 1 \right) -q \frac{d}{dq}\left(\frac{ q^k }{k!^2} (1-2kH_k)\right)\vp_I +S_W\pi_W\left(\vp_P\cdot \frac{-2 q^{k-1} }{(k-1)!^2} H_{k-1}\vp_P\right)\\
& =\frac{q^k}{k!^2}(2k^2 H_k -k) \vp_I-\frac{2 q^{k-1} }{(k-1)!^2} H_{k-1} q\vp_I =\frac{q^k}{k!^2}(2k^2 H_k -k-2k^2 H_{k-1})\vp_I  \\
&= \frac{q^k}{k!^2}(2k -k)\vp_I = \frac{kq^k}{k!^2}\vp_I = \frac{q^k }{k!(k-1)!}\vp_I = C_{2k-1}(P).
\end{split}
\ee
For odd descendants of $I$, we evaluate
\be
\begin{split}
&C_W(\vp_P, C_{2k+1}(I)) -q \frac{d}{dq}C_{2k+1}(I) +S_W\pi_W(\vp_P\cdot C_{2k}(I)) \\
&= C_W\left(\vp_P, -\frac{2 q^k }{k!^2}H_k\vp_P \right) -q \frac{d}{dq}\left(-\frac{2 q^k }{k!^2}H_k \vp_P\right) +S_W\pi_W\left(\vp_P\cdot \frac{q^k}{k!^2}(1-2kH_k)\vp_I \right)\\
& = -\frac{2 q^k }{k!^2}H_k \vp_P+\frac{2(k+1) q^k }{k!^2} H_k \vp_P+\frac{q^k}{k!^2}(1-2kH_k)\vp_P =\frac{q^k}{k!^2}\vp_P = C_{2k}(P).
\end{split}
\ee
\end{proof}
\begin{Remark}
In our work \cite{Losev:2023uxa}, we derived the divisor relation for a particular case of the LGS theory that mirrors the GW theory on the toric surface.
\end{Remark}

\subsection{Topological recursion relation}
\begin{Proposition} For $n\geq 3$ and $m_k\geq 0$ and 
\be\label{eq_gw_trr_via_lgs_trr}
\begin{split}
\< &\Phi_{m_1} (\gamma_1), \Phi_{m_2}(\gamma_2), \ldots,  \Phi_{m_n}(\gamma_n) \>_W = \<\tau_{m_1-1}(\g_1), \g_a\> g^{ab}\< \vp_b, \Phi_{m_2}(\gamma_2), \ldots,  \Phi_{m_n}(\gamma_n) \>_W\\
&+\sum \< \vp_a, \Phi_{m_1-1} (\gamma_1),\{\Phi_{m_i} (\gamma_i)\}_{i\in S_1} \>_W g^{ab} \< \vp_b, \Phi_{m_2}(\gamma_2), \Phi_{m_3}(\gamma_3), \{\Phi_{m_j} (\gamma_j)\}_{j\in S_2}\>_W.
\end{split}
\ee
The sum is taken over possible subsets $S_1\neq \emptyset$ and $S_2$ such that $S_1\cup S_2 = \{4,\ldots,n\}$.
\end{Proposition}
\begin{proof}
The LGS TRR is linear in each argument, while $\Phi_m(\g)$ from (\ref{eq_KM_map}) is a linear combination of LGS descendants. However, the key difference between the LGS TRR and GW TRR is the presence of the 2-point function in the GW TRR. Let us consider 
\be\label{eq_proof_GW_trr_from_LGS}
\begin{split}
\< \Phi_{m_1} (\gamma_1), \Phi_{m_2}(\gamma_2), \ldots, & \Phi_{m_n}(\gamma_n) \>_W =\< z\Phi_{m_1-1} (\gamma_1), \Phi_{m_2}(\gamma_2), \ldots,  \Phi_{m_n}(\gamma_n) \>_W \\
&+\< C_{m_1} (\gamma_1), \Phi_{m_2}(\gamma_2), \ldots,  \Phi_{m_n}(\gamma_n) \>_W .
\end{split}
\ee
We apply the LGS TRR to the first term in the rhs of (\ref{eq_proof_GW_trr_from_LGS}) to get the second line of (\ref{eq_gw_trr_via_lgs_trr}). We use the Kontsevich-Manin representation (\ref{eq_KM_map_2pt_funct}) for the $C_{m_1}(\g_1)$ for the second line (\ref{eq_proof_GW_trr_from_LGS}), so that   
\be\label{eq_second_line_KM_rep}
\< C_{m_1} (\gamma_1), \Phi_{m_2}(\gamma_2), \ldots,  \Phi_{m_n}(\gamma_n) \>_W  = \<\tau_{m_1-1}(\g_1), \g_a\> g^{ab}\< \vp_b, \Phi_{m_2}(\gamma_2), \ldots,  \Phi_{m_n}(\gamma_n) \>_W.
\ee
Indeed (\ref{eq_second_line_KM_rep}) is identical to the first line of the proposition (\ref{eq_gw_trr_via_lgs_trr}), so the proof is complete.
\end{proof}

\section{Selected examples of GW invariants via LGS theory}\label{sect_examples}

This section discusses selected examples of the mirror LGS correlation functions.
\subsection{Four-point functions}
The 4-point correlation functions of the four descendants are well-known in the literature. We use expressions from \cite{norbury2014gromov}. Namely,
\be
\begin{split}
\< \tau_{2m_1}(P) \tau_{2m_2}(P) \tau_{2m_3}(P) \tau_{2m_4}(P) \>& = \frac{1+m_1+m_2+m_3+m_4}{m_1!^2 m_2!^2 m_3!^2 m_4!^2},\\
\< \tau_{2m_1}(P) \tau_{2m_2}(P) \tau_{2m_3-1}(P) \tau_{2m_4-1}(P) \>& = m_3 m_4\frac{m_1+m_2+m_3+m_4}{m_1!^2 m_2!^2 m_3!^2 m_4!^2},\\
\< \tau_{2m_1-1}(P) \tau_{2m_2-1}(P) \tau_{2m_3-1}(P) \tau_{2m_4-1}(P) \>& = m_1 m_2 m_3 m_4\frac{m_1+m_2+m_3+m_4}{m_1!^2 m_2!^2 m_3!^2 m_4!^2}.
\end{split}
\ee
For the 4-point functions, we only need two leading orders in $z$-expansion of (\ref{eq_KM_map}), i.e.
\be
\begin{split}
\Phi_{2m-1}(P) &= \frac{q^m}{m!(m-1)!} \left(\vp_I + z q^{-1}m \vp_P \right)+\cO(z^2), \\
\Phi_{2m}(P) &= \frac{q^m}{m!^2} \left(\vp_P + z m\vp_I \right)+\cO(z^2).
\end{split}
\ee
The LGS correlator of even descendants evaluates to 
\be\nn
\begin{split}
&\<\Phi_{2m_1}(P), \Phi_{2m_2}(P), \Phi_{2m_3}(P), \Phi_{2m_4}(P)\>_W \\
&=\prod_{j=1}^4 \frac{q^{m_j}}{m_j!^2} \big(\<\vp_P, \vp_P, \vp_P, \vp_P\>_W + (m_1+m_2+m_3+m_4)\<z\vp_I, \vp_P, \vp_P, \vp_P\>_W \big) \\
&=\prod_{j=1}^4 \frac{q^{m_j}}{m_j!^2} (q+ (m_1+m_2+m_3+m_4)q ) =\frac{1+m_1+m_2+m_3+m_4}{m_1!^2 m_2!^2 m_3!^2 m_4!^2}q^{m_1+m_2+m_3+m_4+1}.
\end{split}
\ee
We used the extreme LGS correlator from example \ref{ex_4_point_extreme}.

The LGS correlator of mixed descendants evaluates to 
\be\nn
\begin{split}
&\<\Phi_{2m_1}(P), \Phi_{2m_2}(P), \Phi_{2m_3-1}(P), \Phi_{2m_4-1}(P)\>_W =m_3m_4\prod_{j=1}^4 \frac{q^{m_j}}{m_j!^2} \big(\<\vp_P, \vp_P, \vp_I, \vp_I\>_W \\
&+ (m_1+m_2)\<z\vp_I, \vp_P, \vp_I, \vp_I\>_W +q^{-1}(m_3+m_4) \<\vp_P, \vp_P, \vp_I, z\vp_P\>_W \big) \\
&=m_3 m_4\prod_{j=1}^4 \frac{q^{m_j}}{m_j!^2} ((m_1+m_2)+q^{-1}(m_3+m_4)q ) =m_3m_4\frac{m_1+m_2+m_3+m_4}{m_1!^2 m_2!^2 m_3!^2 m_4!^2}q^{m_1+m_2+m_3+m_4}.
\end{split}
\ee
The odd descendants correlator 
\be\nn
\begin{split}
&\<\Phi_{2m_1-1}(P), \Phi_{2m_2-1}(P), \Phi_{2m_3-1}(P), \Phi_{2m_4-1}(P)\>_W \\
&=m_1m_2m_3m_4\prod_{j=1}^4 \frac{q^{m_j}}{m_j!^2} \big(\<\vp_I, \vp_I, \vp_I, \vp_I\>_W +q^{-1} (m_1+m_2+m_3+m_4) \<\vp_I, \vp_I, \vp_I, z\vp_P\>_W \big) \\
&=m_1m_2m_3m_4\frac{m_1+m_2+m_3+m_4}{m_1!^2 m_2!^2 m_3!^2 m_4!^2}q^{m_1+m_2+m_3+m_4-1}.
\end{split}
\ee

\subsection{Selected 5- and 6-point functions}
The level-two descendant GW invariants from Dubrovin-Yang 
\be
\begin{split}
\<\tau_2(P)^5\> &=6^2 q^6 = 36 q^6 ,\;\;\<\tau_2(P)^6\> =7^3 q^7= 343 q^7 ,\\
\<\tau_4(P)^5\> &=\frac{121}{1024} q^{11},\;\; \<\tau_3(P)^6\> = \frac{333}{16}.
\end{split}
\ee
The 5-point function of level-2 descendants
\be\nn
\begin{split}
\< \Phi_{2}(P)^{\otimes 5}\>_W &=q^5\<(\vp_P+z\vp_I+z^2 q^{-1} \vp_P)^{\otimes 5}\>_W \\
&= q^5 \<\vp_P^{\otimes 5}\>_W  + 5q^5 \<\vp_P^{\otimes 4}, z\vp_I\>_W + \binom{5}{2} q^5 \< \vp_P^{\otimes 3},z\vp_I,z\vp_I\>_W + 5 q^4 \<\vp_P^{\otimes 4}, z^2\vp_P\>_W\\
& = q^5 \cdot q  + 2\cdot 5q^5 \<\vp_P^{\otimes 4}\>_W + 10 q^5 \binom{2}{1,1} \oint \frac{\vp_P^3}{W'} + 5 q^4\binom{2}{2}\oint \frac{\vp_P^5}{W'}\\
& = q^5 \cdot q  + 10q^6  + 20 q^6  + 5 q^6 = 36\; q^6.
\end{split}
\ee
We used the the LGS dilaton relation (\ref{eq_LGS_dilaton_relation}) to simplify $\<\vp_P^{\otimes 4}, z\vp_I\>_W $ and critical correlation function formula (\ref{def_extr_LGS}) for $\< \vp_P^{\otimes 3},z\vp_I,z\vp_I\>_W$ and $\<\vp_P^{\otimes 4}, z^2\vp_P\>_W$.

The 6-point function of level-2 descendants
\be\nn
\begin{split}
\<& \Phi_{2}(P)^{\otimes 6}\>_W =\<(q\vp_P+zq\vp_I+z^2\vp_P)^{\otimes 6}\>_W  = q^6 \<\vp_P^{\otimes 6}\>_W  + 6q^6 \<\vp_P^{\otimes 5}, z\vp_I\>_W + 6 q^5 \<\vp_P^{\otimes 5}, z^2\vp_P\>_W \\
&+ \binom{6}{2} q^6 \< \vp_P^{\otimes 4},(z\vp_I)^{\otimes 2}\>_W + \binom{6}{3} q^6 \< \vp_P^{\otimes 3},(z\vp_I)^{\otimes 3}\>_W  +6\cdot 5 q^5 \<\vp_P^{\otimes 4},z \vp_I, z^2\vp_P\>_W \\
& = q^6 \cdot q  + 3\cdot 6 q^6 \<\vp_P^{\otimes 5}\>_W + 15 q^6 \cdot 3\cdot 2  \<\vp_P^{\otimes 4}\>_W + 20 q^6\binom{3}{1,1,1}\oint \frac{\vp_P^3}{W'}+ 6 q^5 \cdot 4q^2 \\
&+ 30 q^5  \binom{3}{2,1} \oint \frac{\vp_P^5}{W'} = q^7 + 18q^7  + 90 q^7  + 120 q^7  + 24q^7 + 90 q^7 = 343\; q^7.
\end{split}
\ee
We used the LGS dilaton relation (\ref{eq_LGS_dilaton_relation}), critical correlation function formula (\ref{def_extr_LGS}) and explicit evaluation of 
\be\nn
\begin{split}
\<\vp_P^{\otimes 5}, z^2\vp_P\>_W &= q\p_q \<\vp_P^{\otimes 4}, z^2\vp_P\>_{W} +  \<\vp_P^{\otimes 4}, z S\pi (\vp_P\vp_P)\>_W\\
& = q\p_q \binom{2}{2} \oint \frac{\vp_P^5}{W'}+  q\<\vp_P^{\otimes 4}, z \vp_I\>_W =  q\p_q q^2 + q \cdot 2 \<\vp_P^{\otimes 4}\>_W = 2q^2 +2q^2= 4q^2.
\end{split}
\ee
The 6-point function of level-3 descendants
\be\nn
\begin{split}
2^6 &\< \Phi_3(P)^{\otimes 6}\>_W = \< \left(q^2\vp_I +2zq\vp_P+2z^2q \vp_I+2z^3 \vp_P\right)^{\otimes6} \> \\
& = q^{12}\<\vp_I^{\otimes 6} \>_W + 6 q^{10}\cdot 2q \< \vp_I^{\otimes 5}, z\vp_P\>_W + 6 q^{10}\cdot 2q \< \vp_I^{\otimes 5},z^2\vp_I\>_W + 6 q^{10}\cdot 2 \< \vp_I^{\otimes 5}, z^3 \vp_P\>_W\\
&+\binom{6}{2} q^{8}\cdot (2q)^2 \< \vp_I^{\otimes 4}, z\vp_P, z\vp_P\>_W+6\cdot 5 q^{8}\cdot (2q)^2 \< \vp_I^{\otimes 4}, z\vp_P, z^2\vp_I\>_W\\
&+\binom{6}{3} q^{6}\cdot (2q)^3 \< \vp_I^{\otimes 3}, z\vp_P, z\vp_P, z\vp_P\>_W\\
& = 0+ 6 q^{10}\cdot 2q \< \vp_I^{\otimes 4}, \vp_P\>_W + 6 q^{10}\cdot 2q \< \vp_I^{\otimes 4}\>_W + 6 q^{10}\cdot 2\cdot \binom{3}{3} \oint \frac{\vp_I^{5}\vp_P}{W'}\\
&+\binom{6}{2} q^{8}\cdot (2q)^2 \cdot 2 \< \vp_I^{\otimes 2}, \vp_P, \vp_P\>_W+6\cdot 5 q^{8}\cdot (2q)^2 \cdot \binom{3}{2,1} \oint \frac{\vp_I^{5}\vp_P}{W'}\\
&+\binom{6}{3} q^{6}\cdot (2q)^3 \cdot \binom{3}{1,1,1} \oint \frac{\vp_I^{3}\vp_P^3}{W'}\\
& = 0+ 0 +0 + 12 q^{10}+0+6\cdot 5 q^{8}\cdot (2q)^2 \cdot 3+\frac{6!}{3!3!} q^{6}\cdot (2q)^3 \cdot  3! q= 1332 q^{10}.
\end{split}
\ee
We used the LGS puncture relation and extreme correlation function formula (\ref{def_extr_LGS}).

The 5-point function of level-4 descendants
\be\nn
\begin{split}
4^5 &\< \Phi_4(P)^{\otimes 5}\>_W = q^{10}\< \left(\vp_P+ 2 z\vp_I+  4z^2q^{-1} \vp_P+ 4z^3 q^{-1} \vp_I+4z^4 q^{-2}\vp_P\right)^{\otimes5} \> \\
& = q^{10}\<\vp_P^{\otimes 5} \>_W + 5 q^{10} \< \vp_P^{\otimes 4}, 2z \vp_I\>_W +  \binom{5}{2}
q^{10} \< \vp_P^{\otimes 3}, 2z \vp_I, 2z \vp_I \>_W+ 5 q^{9}  \<\vp_P^{\otimes 4}, 4z^2 \vp_P\>_W \\
& = q^{10}\cdot q  +10q^{10}\cdot 2 q+ 40 q^{10} \binom{2}{1,1} \oint \frac{ \vp_P^3}{W'} +  20 q^{9} \oint \frac{ \vp_P^5}{W'} \\
& = q^{11} + 20 q^{11} + 80 q^{11} + 20 q^{11}= 121\; q^{11}.
\end{split}
\ee
We used the LGS dilaton relation (\ref{eq_LGS_dilaton_relation}) and critical correlation function formula (\ref{def_extr_LGS}).

\subsection{Even descendants}
The \cite{norbury2014gromov} provided the following formula for the GW invariant for even descendants
\be\label{eq_NS_even_desc}
\left\< \prod_{i=1}^n \tau_{2m_i} (P) \right\> =q \prod_{i=1}^n \frac{q^{m_i}}{m_i!^2} \left( 1+\sum_{i=1}^n m_i\right)^{n-3}.
\ee
Here, we will use the KM mirror map and LGS theory to reproduce several leading terms in (\ref{eq_NS_even_desc}) as the power series expansion over $m_j$. Let us expand the rhs of (\ref{eq_NS_even_desc}) in a power series of $m$
\be
\begin{split}
&\left( 1+\sum_{i=1}^n m_i\right)^{n-3} = \sum_{k=0}^{n-3} \binom{n-3}{k} \left(\sum_{i=1}^n m_i\right)^k  = 1 +(n-3) \sum_{i=1}^n m_i \\&+ \binom{n-3}{2} \sum_{i=1}^n m_i^2 + \binom{n-3}{2} \sum_{i\neq j}m_i m_j+\cO(m^3).
\end{split}
\ee
The KM mirror map (\ref{eq_KM_map}) for the even descendants
\be\nn
\begin{split}
q^{-m}m!^2 \cdot \Phi_{2m} (P) &  = \vp_P + mz \vp_I + m^2 q^{-1}z^2 \vp_P + m^2(m-1) q^{-1}z^3 \vp_I+\ldots. 
\end{split}
\ee
We use the mirror theorem (\ref{thm_main}) to rewrite the GW invariants in (\ref{eq_NS_even_desc}) as LGS correlators. Namely, the expansion up to the total descendant level-3 is
\be\nn
\begin{split}
&\left\< \prod_{i=1}^n q^{-m_i} m_i!^2\Phi_{2m_i} (P) \right\>_W = \< \vp_P^{\otimes n}\>_W + \sum_{i=1}^n m_i \< \vp_P^{\otimes (n-1)}, z\vp_I\>_W \\
&+ \sum_{i<j}^n m_i m_j\< \vp_P^{\otimes (n-2)}, z\vp_I, z\vp_I\>_W + q^{-1}\sum_{i=1}^n m_i^2 \< \vp_P^{\otimes (n-1)}, z^2\vp_P\>_W +\<\cO(z^3)\>_W\\
& = 1 + (n-3)\sum_{i=1}^n m_i q +(n-3)(n-4)q\sum_{i<j}^n m_i m_j+q^{-1} q^2(\cdots) \sum_{i=1}^n m_i^2 + \<\cO(z^3)\>_W.
\end{split}
\ee
We use the LGS dilaton relation (\ref{eq_LGS_dilaton_relation}) to simplify correlators of $z\vp_I$ and correlation functions (\ref{eq_n_point_lgs_points}). We did not evaluate the last term, since the $\sum m_i^2$ type terms will appear in the LGS correlators with higher powers of $z$, since the coefficients in the KM mirror map (\ref{eq_KM_map}) are polynomials in $m$ rather than monomials.

\section{Applications} 
This section describes several universal applications of the mirror LGS description for the GW invariants: The Hurwitz numbers computation, certain polynomiality and integrality properties.

\subsection{Hurwitz numbers}
The GW invariant for the even number of level-1 descendants equals the number of $\mathbb{P}^1$ coverings of genus zero with simple ramification points. Namely, we can use an exact formula for the Hurwitz numbers to express
\be\label{eq_Hurwitz_GW_rep}
\< \tau_1(P)^{2m}\>  =q^{m+1} H_{0,m+1} =  \frac{(2m)!}{(m+1)!} (m+1)^{m-2}q^{m+1}.
\ee
The first several numbers for (\ref{eq_Hurwitz_GW_rep}) are
\be\label{eq_Hurwitz_numerical_values}
\begin{split}
H_{0,3} &= 4,\;\; H_{0,4}  = 120,\;\; H_{0,5}  = 8400,\; H_{0,6} = 1088640.
\end{split}
\ee
We use the KM map (\ref{eq_KM_map}) to express the $n$-point functions of $\tau_1(P)$ in the LGS theory
\be
\begin{split}
\< &\Phi_{1}(P)^{\otimes 2m}\>_W =\<(q\vp_I+z\vp_P)^{\otimes 2m}\>_W   = \sum_{k=3}^{m} \binom{2m}{k}  q^k \<\vp_I^{\otimes k} \vp_P^{\otimes (2m-k)}\>_W\\
&= \sum_{k=3}^{m} \binom{2m}{k} \frac{(2m-k)!}{(2m-2k)!} q^k \<\vp_P^{\otimes k}, (z\vp_P)^{\otimes 2(m-k)}\>_W = \sum_{k=3}^{m}  \frac{(2m)!}{k!(2m-2k)!} q^k h_{k, m-k}
\end{split}
\ee
We dropped the terms with $k = 0,1$ since the total descendant level is $2m-k$, which is too big, $n-3 = 2m-3 < 2m-k$. We dropped the terms with $k>m$ since the number of identity insertions $k$ is bigger than that of LGS descendants, which is $2m-k$. Hence, the $2m-k$ iteration of the puncture relation will give us an LGS amplitude with $2m-k$ insertions of $\varphi_P$ and $2(k-m)$ identity insertions. Such amplitude vanishes, except in the special case $\<1,1,\varphi_P\>_W =q$. We have an even number of $\varphi_{P}$, so the special case is excluded. Let us introduce 
\be
h_{k,n}= \<\vp_P^{\otimes k}, (z\vp_P)^{\otimes 2n}\>_W.
\ee
We use the divisor relation to get a recursive formula 
\be
\begin{split}
h_{k,n}&= \<\vp_P^{\otimes k}, (z\vp_P)^{\otimes 2n}\>_W = \<\vp_P,\vp_P^{\otimes (k-1)}, (z\vp_P)^{\otimes 2n}\>_W\\
& = q\p_q \<\vp_P^{\otimes (k-1)}, (z\vp_P)^{\otimes 2n}\>_W + 2n \<\vp_P^{\otimes (k-1)}, (z\vp_P)^{\otimes (2n-1)}, S_W\pi_W (\vp_P\vp_P)\>_W\\
& = (n+1) \<\vp_P^{\otimes (k-1)}, (z\vp_P)^{\otimes 2n}\>_W + 2n q \<\vp_P^{\otimes (k-1)}, (z\vp_P)^{\otimes (2n-1)}, \vp_I\>_W\\
&= (n+1) h_{k-1,n} + 2n(2n-1) q h_{k,n-1}.
\end{split}
\ee
The boundary conditions are the no-descendant and over-extreme cases
\be
h_{k,0} =  \<\vp_P^{\otimes n}\>_W=q,\;\; h_{2,n} =0.
\ee
Note that $h_{3,n}$ is the extreme case 
\be
h_{3,n}= \<\vp_P^{\otimes 3}, (z\vp_P)^{\otimes 2n}\>_W = \binom{2n}{1,\ldots,1} \oint \frac{\vp_P^{2n+3}}{W'} = (2n)! q^{n+1}.
\ee
that matches the recursion 
\be
h_{3,n} = (n+1)\; h_{2,n} + 2n(2n-1) q\; h_{3,n-1} = 2n(2n-1) q\; h_{3,n-1}.
\ee
The final formula for the Hurwitz numbers
\be\nn
\begin{split}
q^{m+1} H_{0,m+1} &= \< \Phi_{1}(P)^{\otimes 2m}\>_W  = \sum_{k=3}^{m}  \frac{(2m)!}{k!(2m-2k)!} q^k h_{k, m-k},\\
h_{k,n}& =(n+1) h_{k-1,n} + 2n(2n-1) q h_{k,n-1},\;\;h_{k,0}=q,\;\;h_{2,n} =0.
\end{split}
\ee
The GW invariant for $m=3$
\be\nn
\begin{split}
\< &\Phi_{1}(P)^{\otimes 6}\>_W=\frac{6!}{3!} q^3 h_{3, 0} = 120\; q^4.
\end{split}
\ee
The GW invariant for $m=4$
\be\nn
\begin{split}
\< &\Phi_{1}(P)^{\otimes 8}\>_W=\frac{8!}{3!2!} q^3 h_{3, 1}+ \frac{8!}{4!} q^4 h_{4, 0} =
\frac{8!}{3!2!} q^3 \cdot 2! q^2+ \frac{8!}{4!} q^4 \cdot q   = 8400\; q^5.
\end{split}
\ee
The GW invariant for $m=5$
\be\nn
\begin{split}
\< &\Phi_{1}(P)^{\otimes 10}\>_W=\frac{10!}{3!4!} q^3 h_{3, 2}+\frac{10!}{4!2!} q^4 h_{4, 1}+ \frac{10!}{5!} q^5 h_{5, 0} \\
& = \frac{10!}{3!4!} q^3 \cdot 4! q^3+\frac{10!}{4!2!} q^4 \cdot 6 q^2 +\frac{10!}{5!} q^5 \cdot q   = 1088640\; q^6.
\end{split}
\ee
We separately evaluated 
\be\nn
h_{4,1} = 2 h_{3,1} + 2q h_{4,0} = 2\cdot 2! q^2 + 2q\cdot q = 6 q^2.
\ee
We observe the perfect matching between the Hurwitz numbers in (\ref{eq_Hurwitz_numerical_values}) and the LGS computations.

\subsection{Polynomiality}
The GW invariants from examples in section \ref{sect_examples} have a structure of a simple polynomial in the descendant levels divided by a product of factorials. Norbury and Scott in \cite{norbury2014gromov} formalized this observation and even provided a proof based on the topological recursion presentation for the GW invariants. In our proof, we use the mirror LGS theory. 
\begin{Theorem} {\bf (Norbury-Scott)} For $g=0$ and even $k$ the GW invariants of $\mathbb{P}^1$ are of the form
\be
\left\< \prod_{i=1}^k \tau_{2m_i-1}(P) \prod_{i=k+1}^n \tau_{2m_k} (P) \right\> = \frac{m_{1}\cdots m_k}{m_1!^2 \cdots m_n!^2} p^g_{n,k} (m_1, \ldots, m_n)\; q^{1-k/2+\sum m_i}.
\ee
Here $p^g_{n,k}$ is a polynomial of degree $3g-3+n$ in $m_i$, symmetric in the first $k$ and last $n-k$ variables, with top coefficient $c_\beta$ of $m_1^{\beta_1}\cdots m_n^{\beta_n}$, for $|\beta| = 3g-3+n\geq 0$ given by 
\be
c_\beta=2^g \int_{\overline{\mathcal{M}_{g,n}}} \psi_1^{\beta_1}\cdots \psi_n^{\beta_n}.
\ee
\end{Theorem}
\begin{proof} We rewrite the mirror map (\ref{eq_KM_map}) for point descendants in the following form
\be\label{eq_polyn_mirror_map}
\begin{split}
m!^2 \;\Phi_{2m}(P)&=\vp_P \sum_{k=0}^m q^{m-k} z^{2k}P_{2k}(m)  + \vp_I \sum_{k=0}^{m-1} q^{m-k} z^{2k+1}P_{2k+1}(m),\\
m!(m-1)!\;\Phi_{2m-1}(P)&=\vp_P \sum_{k=1}^{m} q^{m-k} z^{2k-1}Q_{2k-1}(m)  + \vp_I \sum_{k=0}^{m-1} q^{m-k} z^{2k}Q_{2k}(m). 
\end{split}
\ee
In (\ref{eq_polyn_mirror_map}) introduced polynomials 
\be\label{eq_polynomials_desc}
\begin{split}
P_{2k}(m) &= \prod_{j=0}^{k-1} (m-j)^2,\;\; P_{2k+1}(m)= m \prod_{j=0}^{k-1} (m-j)(m-1-j),\\
Q_{2k}(m)&  = \prod_{j=0}^{k-1} (m-j)(m-j-1),\;\; Q_{2k+1}(m) = m\prod_{j=0}^{k-1} (m-j-1)^2.
\end{split}
\ee
Moreover, in the expansion (\ref{eq_polyn_mirror_map}), the degree of the polynomials is identical to the power of $z$, the LGS descendant level. According to the definition of the LGS correlation functions, the correlation functions of $n$ observables vanish when the total descendant level is $n-2$ or more. The total descendant level is the total $z$-degree, and the corresponding LGS correlator is multiplied by the polynomial of $m_j$ of the same degree. Hence, the maximal degree of the polynomial is $n-3$, identical to the theorem's prediction at $g=0$. 

The top-degree polynomial contributions are multiplied by the extreme correlators. Moreover, the top degree term in polynomials (\ref{eq_polynomials_desc}) 
\be
\begin{split}
P_{a}(m) &= m^{a} + \cO(m^{a-1}),\;\;Q_{a}(m) = m^{a} + \cO(m^{a-1}).
\end{split}
\ee
Hence, the leading monomials with $|\beta| = \b_1+\ldots +\b_n  = n-3$ are of the form  
\be\label{eq_leading_pol_term_point_desc}
m_1^{\b_1}\cdots m_n^{\beta_{n}}\< z^{\b_1} \vp_{\a_1},\ldots, z^{\beta_n} \vp_{\a_n}\>_W = m_1^{\b_1}\cdots m_n^{\beta_{n}} \binom{n-3}{\beta_1,\ldots,\beta_n}
\oint \frac{\vp_{\a_1}\cdots \vp_{\a_n}}{W'}.
\ee
The labels $\a_k \in \{I, P\}$ are such that the residue integral (\ref{eq_residue_good_sections_p1}) of good sections is $1=2^g =2^0$ times the appropriate power of $q$.
The multinomial factor in (\ref{eq_leading_pol_term_point_desc}) is identical to the moduli space integral (\ref{eq_multinomial_moduli_space_integral}).
\end{proof}

\subsection{Integrality and positivity}
The authors of \cite{dubrovin2019gromov} observed the surprising integrality property of the $\mathbb{P}^1$ GW invariants. Namely, the denominators of the descendant invariants typically contain very few prime factors, albeit in high powers. In this section, we explain this observation using the integrality of the LGS invariants. 
\begin{Proposition}\label{prop_lgs_integr}
{\bf (LGS integrality)} For $n+k\geq 3$ and $m_i, l_j \geq 0$
\be
q^{-N}\<z^{m_1}\vp_P,\ldots, z^{m_n}\vp_P, z^{l_1}\vp_I,\ldots, z^{l_k}\vp_I\>_W \in \mathbb{Z}^{\geq 0}. 
\ee
For $N =1+\frac12 (m_1+\ldots+m_n+l_1+\ldots+l_k-k)$. When $N$ is a half-integer, the correlation function vanishes.
\end{Proposition}
\begin{proof}
The LGS correlation function for the mirror theory to $\mathbb{P}^1$ is defined recursively via puncture and divisor relations, starting with the extreme case. The extreme correlator (\ref{def_extr_LGS}) is a positive integer multiple of the residue. The residue (\ref{eq_residue_good_sections_p1}) is a positive integer multiple of $q$-power. Hence, the non-zero LGS correlation function for the descendants of good sections is a positive integer.
\end{proof}
The descendant GW invariants are rational numbers. The mirror map (\ref{eq_KM_map}) implies that the mirrors of descendants have very simple common denominators. Hence, we can formulate the following integrality properties
\begin{Theorem}{\bf (GW integrality)}
For $k+l \geq 3$ and $m_j\geq 0$ and $q=1$ the GW descendant invariants are integer numbers divided by the products of factorials. Namely,
\be\nn
\begin{split}
 &\left\< \prod_{i=1}^k m_i!^2\tau_{2m_i}(P) \prod_{j=1}^l m_j!(m_j+1)!\tau_{2m_j+1}(P)\right\> \in \mathbb{Z}^{\geq 0},\\
 &\left\< \prod_{i=1}^n m_i!^3\tau_{2m_i+1}(I) \prod_{j=1}^p m_j^2!(m_j-1)!\tau_{2m_j}(I)\prod_{i=1}^k m_i!^2\tau_{2m_i}(P) \prod_{j=1}^l m_j!(m_j+1)!\tau_{2m_j+1}(P)\right\> \in \mathbb{Z}.
\end{split}
\ee
\end{Theorem}\label{thm_integrality}
\begin{proof} The coefficients in the descendant expansion (\ref{eq_polyn_mirror_map}), the values of polynomials (\ref{eq_polynomials_desc}) are integer numbers. Namely,
\be\label{eq_integr_mirror_map}
\begin{split}
m!^2 \cdot \Phi_{2m} (P) & = \sum_{k=0}^{2m} z^k a_{k} \vp_k,\;\; m!(m+1)! \cdot \Phi_{2m+1} (P)= \sum_{k=0}^{2m+1} z^k b_{k} \vp_k,\\
& a_k, b_k \in \mathbb{N},\;\; \vp_k \in\{\vp_I, \vp_P\}.
\end{split}
\ee
We use the theorem \ref{thm_main} to express the GW invariants of point descendants as a linear combination of LGS correlation functions. The structure of the mirror map (\ref{eq_integr_mirror_map}) implies that the coefficients in the linear combination are positive integers. The LGS integrality proposition \ref{prop_lgs_integr} implies that the LGS correlation functions are non-negative integers. Hence, the GW invariants of point descendants, multiplied by the corresponding factorials, are non-negative integer numbers. 

The proof of the second statement of the integrality theorem \ref{thm_integrality} is similar. The descendants (\ref{eq_KM_map}) of identity class include harmonic numbers $H_k$ in the KM mirror map. Harmonic numbers satisfy 
\be
m! \cdot H_k \in \mathbb{Z}^{>0},\;\; \forall k \leq m. 
\ee
We can turn all coefficients into integer numbers if we multiply by an additional factorial factor. Namely, we can express the descendants of the identity class in the form
\be\label{eq_integr_mirror_map_identity}
\begin{split}
m!^2(m-1)! &\cdot \Phi_{2m}(I) = \sum_{k=0}^{2m} z^k e_{k} \vp_k,\;\; m!^3 \cdot \Phi_{2m+1} (I)= \sum_{k=0}^{2m+1} z^k f_{k} \vp_k,\\
& e_k, f_k \in \mathbb{Z},\;\; \vp_k \in\{\vp_I, \vp_P\}.
\end{split}
\ee
\end{proof}
Below, we provide an example of the LGS correlation function of the level-2 identity descendant with different signs, depending on other observables in the correlation function.
\begin{Example}
We use the descendant map (\ref{eq_low_identity_desc_mirr_map}) for the level-2 descendant of identity to evaluate
\be
\begin{split}
\<\Phi_2(I), z\vp_P, \vp_P^{\otimes 4}\>_W &=  \<z^2\vp_I, z\vp_P,\vp_P^{\otimes 4}\>_W-\<q\vp_I, z\vp_P,\vp_P^{\otimes 4}\>_W \\
&=\binom{3}{2,1} \oint \frac{\vp^5_P}{W'}-q \<\vp_P^{\otimes 5}\>_W  = 3q^2-q^2=2q^2, \\
\<\Phi_2(I), z\vp_P, \vp_P, \vp_P\>_W &= -q \<\vp_I, z\vp_P,\vp_P,\vp_P\>_W =-q \<\vp_P, \vp_P, \vp_P\>_W =-q^2.
\end{split}
\ee
\end{Example}

\section*{Acknowledgments}

The author is grateful to Andrey Losev for discussing the topics presented in this paper and to Andrey Okounkov for his comments and feedback. The author is grateful to A. Alexandrov, P. Dunin-Barkowski, M. Karev, M. Kazarian, and Di Yang for valuable discussions and the organizers of the workshop "Noncommutative Geometry Meets Topological Recursion" for providing a space for the talks.

\bibliography{LGS_P1_ref}{}

\providecommand{\href}[2]{#2}\begingroup\raggedright\begin{thebibliography}{10}

\bibitem{dubrovin2019gromov}
B.~Dubrovin and D.~Yang, ``On Gromov--Witten invariants of P1,'' {\em Math.
  Res. Lett} {\bfseries 26} no.~3, (2019) 729--748.

\bibitem{pandharipande2000toda}
``The Toda equations and the Gromov--Witten theory of the Riemann sphere,''
  {\em Letters in Mathematical Physics} {\bfseries 53} (2000) 59--74.

\bibitem{okounkov2006gromov}
A.~Okounkov and R.~Pandharipande, ``Gromov-Witten theory, Hurwitz theory, and
  completed cycles,'' {\em Annals of mathematics} (2006) 517--560.

\bibitem{saito1983period}
K.~Saito, ``Period mapping associated to a primitive form,'' {\em Publications
  of the Research Institute for Mathematical Sciences} {\bfseries 19} no.~3,
  (1983) 1231--1264.

\bibitem{Losev:1992tt}
A.~Losev, ``{Descendants constructed from matter field in topological
  Landau-Ginzburg theories coupled to topological gravity},''
  \href{http://dx.doi.org/10.1007/BF01017145}{{\em Theor. Math. Phys.}
  {\bfseries 95} (1993) 595--603},
  \href{http://arxiv.org/abs/hep-th/9211090}{{\ttfamily arXiv:hep-th/9211090}}.

\bibitem{eguchi1993topological}
T.~Eguchi, H.~Kanno, Y.~Yamada, and S.-K. Yang, ``Topological strings, flat
  coordinates and gravitational descendants,'' {\em Physics Letters B}
  {\bfseries 305} no.~3, (1993) 235--241.

\bibitem{losev1995connection}
A.~Losev and I.~Polyubin, ``On connection between topological Landau-Ginzburg
  gravity and integrable systems,'' {\em International Journal of Modern
  Physics A} {\bfseries 10} no.~29, (1995) 4161--4178.

\bibitem{eguchi1995topological}
T.~Eguchi, K.~Hori, and S.-K. Yang, ``Topological $\sigma$ models and large N
  matrix integral,'' {\em International Journal of Modern Physics A} {\bfseries
  10} no.~29, (1995) 4203--4224.

\bibitem{Eguchi:1996tg}
T.~Eguchi, K.~Hori, and C.-S. Xiong, ``{Gravitational quantum cohomology},''
  \href{http://dx.doi.org/10.1142/S0217751X97001146}{{\em Int. J. Mod. Phys. A}
  {\bfseries 12} (1997) 1743--1782},
  \href{http://arxiv.org/abs/hep-th/9605225}{{\ttfamily arXiv:hep-th/9605225}}.

\bibitem{givental2001gromov}
A.~B. Givental', ``Gromov--Witten invariants and quantization of quadratic
  Hamiltonians,'' {\em Moscow Mathematical Journal} {\bfseries 1} no.~4, (2001)
  551--568.

\bibitem{takahashi1998primitive}
A.~Takahashi, ``Primitive Forms, Topological LG models coupled to Gravity and
  Mirror Symmetry,'' 1998.
\newblock \url{https://arxiv.org/abs/math/9802059}.

\bibitem{kontsevich1998relations}
M.~Kontsevich and Y.~Manin, ``Relations between the correlators of the
  topological sigma-model coupled to gravity,'' {\em Communications in
  Mathematical Physics} {\bfseries 196} no.~2, (1998) 385--398.

\bibitem{norbury2014gromov}
P.~Norbury and N.~Scott, ``Gromov--Witten invariants of P1 and Eynard--Orantin
  invariants,'' {\em Geometry \& Topology} {\bfseries 18} no.~4, (2014)
  1865--1910.

\bibitem{dunin2014identification}
P.~Dunin-Barkowski, N.~Orantin, S.~Shadrin, L.~Spitz, {\em et al.},
  ``Identification of the Givental formula with the spectral curve topological
  recursion procedure,'' {\em Communications in Mathematical Physics}
  {\bfseries 328} (2014) .

\bibitem{manin1999frobenius}
Y.~Manin, {\em Frobenius manifolds, quantum cohomology, and moduli spaces},
  vol.~47.
\newblock American Mathematical Soc., 1999.

\bibitem{witten1990two}
E.~Witten, ``Two-dimensional gravity and intersection theory on moduli space,''
  {\em Surveys in differential geometry} {\bfseries 1} no.~1, (1990) 243--310.

\bibitem{Losev:1994whn}
A.~Losev, ``{Structures of K. Saito theory of primitive form in topological
  theories coupled to topological gravity},''
  \href{http://dx.doi.org/10.1007/3-540-58453-6_10}{{\em Lect. Notes Phys.}
  {\bfseries 436} (1994) 172--193}.

\bibitem{Losev:1998dv}
A.~Losev, ``{'Hodge strings' and elements of K. Saito's theory of the primitive
  form},'' in {\em {Taniguchi Symposium on Topological Field Theory, Primitive
  Forms and Related Topics}}, pp.~305--335.
\newblock 1, 1998.
\newblock \href{http://arxiv.org/abs/hep-th/9801179}{{\ttfamily
  arXiv:hep-th/9801179}}.

\bibitem{Losev:2022tzr}
A.~Losev and V.~Lysov, ``{Tropical Mirror},''
  \href{http://dx.doi.org/10.3842/SIGMA.2024.072}{{\em SIGMA} {\bfseries 20}
  (2024) 072}, \href{http://arxiv.org/abs/2204.06896}{{\ttfamily
  arXiv:2204.06896 [hep-th]}}.

\bibitem{losev2023tropical}
A.~Losev and V.~Lysov, ``{Tropical Mirror Symmetry: Correlation functions},''
  \href{http://arxiv.org/abs/2301.01687}{{\ttfamily arXiv:2301.01687
  [hep-th]}}.

\bibitem{Losev:2023uxa}
A.~Losev and V.~Lysov, ``{Tropical mirror for toric surfaces},''
  \href{http://arxiv.org/abs/2305.00423}{{\ttfamily arXiv:2305.00423
  [hep-th]}}.

\end{thebibliography}\endgroup
\bibliographystyle{utphys}

\end{document}